\documentclass[11pt,a4paper,citesort,svgnames,dvipsnames]{article}
\usepackage{amsfonts,amssymb,amsmath,amsthm,amscd,latexsym}
\DeclareMathOperator{\sech}{sech}
\usepackage[latin1]{inputenc}
\usepackage{graphicx}
\usepackage[all]{xy}
\usepackage{caption}
\usepackage{subcaption}
\usepackage{lmodern}
\usepackage{microtype}
\usepackage{verbatim}
\usepackage{indentfirst}
\usepackage{color}
\usepackage[official]{eurosym}

\usepackage{color}
\usepackage{marginnote} 
\usepackage{ifthen} 
\usepackage{tensor}
\usepackage{tikz}
\usetikzlibrary{matrix,arrows}
\usepackage[numbers,sort&compress]{natbib}
\usepackage{xcolor}
\usepackage{hyperref}
\hypersetup{
colorlinks,
citecolor=WildStrawberry,
linkcolor=Blue, 
urlcolor=Green}

\def\Ad{\mathrm{Ad}}

\newcommand{\inv}[0]{{-1}}

\newcommand{\oo}[0]{\otimes}

\newcommand{\gotha}{\mathfrak a }
\newcommand{\gothg}{\mathfrak g }

\newcommand{\gothh}{\mathfrak h }

\newcommand{\RR}{\mathbb{R}}
\newcommand{\CC}{\mathbb{C}}

\newcommand{\al}{\alpha}
\newcommand{\bb}{\beta}
\newcommand{\ga}{\gamma}

\newtheorem{theorem}{Theorem}[section]
\newtheorem{lemma}[theorem]{Lemma}

\newtheorem{definition}[theorem]{Definition}

\newcommand{\issue}[2][]{%
  \textcolor{red}{#2}%
  \marginnote{\textbf{\textcolor{red}{issue}}%
    \ifthenelse{\equal{#1}{}}{}{: #1}}}

\def\g                {{\mathfrak g}}

\def\C                {\mathbb C}

\def\id               {{\rm id}}

\def\kk               {{\rm k}}

\long\def\labl#1      {\label{#1}\ee}

\def\R                {\mathbb R}

\def\Z                {\mathbb Z}

\def\k{\kappa}                  
\def\kk{\omega}

\def\conm#1#2{\left[ #1,#2 \right]}  
\def\pois#1#2{\left\{ #1,#2 \right\}}  
\def\1{\'{\i}}

\def\hbar{\mathchar '26\mkern -9muh}

\def\CC{\text{\ \!C}}       
        
\def\SS{\text{\ \!S}}        

\def\ea{{\rm e}}

 \def\m{{\eta}}  

\def\C{{\Upsilon}}

\def\back{\!\!\!\!\!\!\!}            
\def\mback{\!\!\!\!}

\newcommand\be{\begin{equation}}
\newcommand\ee{\end{equation}}
\newcommand\bea{\begin{eqnarray}}
\newcommand\eea{\end{eqnarray}}

\advance \voffset by -.5in
\advance \hoffset by -.75in
\begin{document}

\ \bigskip\bigskip

\begin{center}
\baselineskip 24 pt {\LARGE \bf  
AdS Poisson homogeneous spaces and \\
Drinfel'd doubles}

\end{center}

 \medskip

\begin{center}

{\sc Angel  Ballesteros$^1$, Catherine Meusburger$^2$, Pedro Naranjo$^1$}

{$^1$ Departamento de F\1sica, Universidad de Burgos, 
E-09001 Burgos, Spain}

{$^2$  Department Mathematik,  FAU Erlangen-N\"urnberg, Cauerstr.~11, D-91058 Erlangen, Germany
}
 
e-mail: {angelb@ubu.es, catherine.meusburger@math.uni-erlangen.de, pnaranjo@ubu.es}

\end{center}

\begin{abstract}
The correspondence between Poisson homogeneous spaces over a Poisson-Lie group $G$ and Lagrangian Lie subalgebras of  the classical double $D(\gothg)$ is revisited and explored in detail for the case in which $\gothg=D(\mathfrak a)$ is a classical double itself. We apply these results to give an explicit  description of some coisotropic 2d Poisson homogeneous spaces over the group $\mathrm{SL}(2,\R)\cong\mathrm{SO}(2,1)$, namely 2d anti de Sitter space, 2d hyperbolic space and the lightcone in 3d Minkowski space.  We show how each of these spaces is obtained as 
a quotient with respect to a Poisson-subgroup for  one of the three inequivalent Lie bialgebra structures on $\mathfrak{sl}(2,\R)$ and as a coisotropic one for the others. 

We then construct families  of  coisotropic Poisson homogeneous structures for 3d anti de Sitter space $\mathrm{AdS}_3$ and show that the ones that are quotients by a Poisson subgroup are determined by a three-parameter family of classical $r$-matrices for $\mathfrak{so}(2,2)$, while the {non Poisson-subgroup} cases are much more numerous. In particular, we present the two Poisson homogeneous structures on $\mathrm{AdS}_3$ that arise from two Drinfel'd double structures on  $\mathrm{SO}(2,2)$. The first one realises $\mathrm{AdS}_3$  as a quotient of $\mathrm{SO}(2,2)$ by the Poisson-subgroup
$\mathrm{SL}(2,\R)$, while the second one,   the non-commutative spacetime of the twisted $\kappa$-AdS deformation, realises $\mathrm{AdS}_3$ as a  coisotropic Poisson homogeneous space. 
\end{abstract}

\medskip 

\noindent




\section{Introduction}

Non-commutative spacetimes are widely believed to provide a suitable framework for the description of fundamental properties of spacetime when quantum gravity effects are taken into consideration, see for instance~\cite{Snyder, Woronowicz, Maggiore, FredenCMP, CKNT} and references therein.
Mathematically, they are given as  (co)module (co)algebras over  quantum groups and the associated sub-(co)algebras of (co)invariants. 
A prominent example are the $q$-Minkowski spacetimes arising from quantum Poincar\'e symmetries whose  properties have been thoroughly studied in~\cite{kMas, kMR, kZakr,kappaP,nullplane}. However, much less is known about their  counterparts for non-vanishing cosmological constant, quantum de Sitter and anti de Sitter space (see~\cite{kappa31} and references therein). 
This is regrettable since they could serve as  mathematical models of quantum spacetimes with non-vanishing curvature and with possible cosmological implications, see~\cite{BHBruno, Starodutsev, Marciano, iceCUBE}.

From a classical perspective, it is well known that $n$-dimensional  Minkowski, de Sitter and anti de Sitter spacetimes are all homogeneous spaces: they are given as quotients $G/H$  of their their isometry group $G$ with respect to their isotropy subgroup $H=\mathrm{SO}(n-1,1)$. 
It is therefore natural to investigate the associated quantum homogeneous spacetimes, which should arise from the associated quantum groups $G_q$  together with a quantum analogue $H_q$ of their isotropy subgroup.  We recall that in the quantum group setting the relevant notion is the Hopf algebra of (non-commutative) functions on the quantum homogeneous space, and all its defining classical properties (transitivity of the action, invariance of a point with respect to a subgroup) have generalisations in the quantum group framework. 
A detailed presentation on quantum homogeneous spaces can be found, for instance, in~\cite{Koor}, where it is stressed that imposing $H_q$ to be a quantum subgroup, i.~e.~a Hopf subalgebra  $H_q\subset G_q$, turns out to be too restrictive. In order to obtain a quantum algebra  that can be viewed as the quantum counterpart of the algebra of functions  on $G/H$  one needs to consider also the {so-called {\em coisotropic}} subgroups since in many relevant 
quantum groups $G_q$ the quantisation of the subgroup $H\subset G$  is not a quantum subgroup. This foundational  issue has been treated in a number of works such as~\cite{EK, Brzezinski, NT, CY1, CY2, Yamashita, KST}, and several examples of quantum homogeneous spaces have been explicitly constructed, see~\cite{Podles, NM, VS, Podles9, Sheu, Ciccoliqplanes, Leit, BCGSTqse, CHZ, HMS, BRV, Tomatsu} and references therein.

However, Lorentzian quantum homogeneous spacetimes in (2+1) and (3+1) dimensions have not been constructed yet -probably due to their non-compact nature and mathematical complexity- although this  would be highly relevant for applications  in quantum gravity. 
At this point, it is worth recalling that in the same manner as quantum groups can be thought of as Hopf algebra counterparts of Poisson-Lie groups, quantum homogeneous spaces can be understood as quantisations of Poisson homogeneous spaces, for an overview see~\cite{Ciccoliqplanes, Zakrzewski, Zakrr1+1, Reyman, Karolinski, Ludrm, Lu, Ciccoli, KS, BCGST, CGDoc} and references therein.  
While  much  simpler on a computational level, these Poisson homogeneous spaces still carry relevant information about the associated quantum homogeneous spaces and can be viewed as their semiclassical limits. 
{In this correspondence, {\em coisotropic} quantum homogeneous spaces  correspond to  {\em coisotropic} Poisson homogeneous spaces.} 

The corresponding structures on the Lie bialgebra level  have been identified in \cite{DrHS} as Lagrangian Lie subalgebras of the classical double $D(\gothg)$. More precisely, it was shown in \cite{DrHS} that the Poisson homogeneous structures on a quotient $G/H$ of a Poisson-Lie group $G$ by a Lie subalgebra $H\subset G$ correspond to orbits of a natural $G$-action  on the algebraic variety of Lagrangian subalgebras of the classical double $D(\gothg)$.  
They 
 can be viewed as the linear or  lowest order approximation of a quantum homogeneous space {at a fixed point} and contain all essential algebraic information for the construction of the  quantum homogeneous space.  Hence, they provide a framework in 
 which  essential properties of Lorentzian quantum homogeneous spacetimes  can be investigated, thus opening the path for their explicit computations in  physically meaningful cases.  {In this correspondence,  {\em coisotropic} Lagrangian subalgebras describe {\em coisotropic} Poisson homogeneous spaces}

In addition to their physics applications in four dimensions, quantum homogeneous spaces and their semiclassical counterparts are also relevant in the context of 3d gravity. On one hand, 3d gravity is widely used as a toy model for the higher-dimensional case and as a testing ground for quantisation approaches. On the other hand, the role of quantum group symmetries is much more transparent in this setting since they arise as the quantum counterparts of Poisson-Lie symmetries in the classical theory.  In particular, it was shown in \cite{bn,br,cmt}  that gauge fixing procedures and the introduction of an observer into the theory lead to the appearance of  dynamical classical $r$-matrices and hence to  Poisson homogeneous spaces.

In  this paper we pursue this  ``first order approach''  to  low dimensional {coisotropic} quantum homogeneous spacetimes, starting with the  Lie bialgebra  $(\gothg,\delta)$ and the  Poisson-Lie group $G$ for quantum group of symmetries   $G_q$. This infinitesimal approach turns out to be the simplest framework in order to understand the plurality of {coisotropic} Poisson homogeneous spaces and hence of {coisotropic} quantum homogeneous spaces that correspond to a given homogeneous space $G/H$. We  examine under this  perspective the construction of  {coisotropic}  2d Poisson homogeneous spaces associated  with the Lorentz group $\mathrm{SO}(2,1)\cong \mathrm{SL}(2,\R)$ in three dimensions  as well as the construction of three-dimensional anti de Sitter space $\mathrm{AdS}_3$ as a {coisotropic} Poisson homogeneous space over its isometry group $\mathrm{SO}(2,2)$.

We show that the three inequivalent 2d homogeneous spaces over $\mathrm{SO}(2,1)\cong\mathrm{SL}(2,\R)$,  2d hyperbolic space, anti de Sitter space and the lightcone in 3d Minkowski space, naturally
correspond  to the three inequivalent bialgebra structures on $\mathfrak{sl}(2,\R)$.  The associated Poisson-Lie structures  allow one to realise one of them as a Poisson homogeneous  space $\mathrm{SL}(2,\R)/H$ with respect to a Poisson-subgroup $H\subset \mathrm{SL}(2,\R)$, while the other two are coisotropic. In all cases, we give an explicit  parametrisation of the Poisson homogeneous spaces in coordinates adapted to their geometry. 

In the case of 3d anti de Sitter space 
$\mathrm{AdS}_3=\mathrm{SO}(2,2)/\mathrm{SO}(2,1)\cong \mathrm{SL}(2,\R)\times\mathrm{SL}(2,\R)/\mathrm{SL}(2,\R)$,
the relevant isotropy subgroup is the group $\mathrm{SL}(2,\R)$,  and there are numerous Poisson-Lie structures on the group $\mathrm{SO}(2,2)$
that give $\mathrm{AdS}_3$ the structure of a {coisotropic} Poisson homogeneous space over $\mathrm{SO}(2,2)$. 
We  consider two representative Lie bialgebra structures on $\mathfrak{so}(2,2)$ that are canonical in the sense that they are the two possible
realisations of $\mathfrak{so}(2,2)$ as a classical double
$\mathfrak{so}(2,2)=D(\mathfrak a)$~\cite{BHMplb1, BHMcqg, BHMplb2, BHMNsigma}. This choice is motivated by
their applications in 3d gravity, where Poisson-Lie structures arise naturally from the description of 3d gravity as a Chern-Simons gauge theory \cite{Witten,AT}  whose symplectic structure can be described in terms of Poisson-Lie structures \cite{FR,AM}. It is argued in \cite{cm2, bernd1} that the natural Poisson-Lie structures for 3d gravity are classical doubles.  For both of these Lie bialgebra structures, we construct the associated Poisson homogeneous space and give an explicit description in coordinates. In the first case, we obtain one of very few descriptions of anti de Sitter space as  a Poisson homogeneous space of the Poisson subgroup type, while the second, that corresponds to the
 twisted $\kappa$-deformation,  is a coisotropic Poisson homogeneous space. 

The structure of the paper is the following. In Section 2 we review the construction of Poisson homogeneous spaces over a Poisson-Lie group $G$ and Drinfel'd's {description} of Poisson homogeneous spaces \cite{DrHS} in terms of Lagrangian Lie  subalgebras  of the classical $D(\gothg)$. 
We give an explicit description in terms of a basis, comment on the special cases of coisotropic and Poisson-subgroup homogeneous spaces  and  explore in detail the case where $\gothg$ is itself a classical double $\gothg=D(\mathfrak a)$ and the homogeneous space is $G/A$.
In Section 3 we construct {coisotropic} Poisson homogeneous spaces over the  group $\mathrm{SL}(2,\R)$ with respect to the three inequivalent Poisson-Lie structures on $\mathrm{SL}(2,\R)$. We {derive} an explicit description of the Poisson structure in terms of coordinates, analyse the structure of the Poisson brackets and their first order approximation. In Section 4 we construct three-dimensional anti de Sitter space as a {coisotropic} Poisson homogeneous space over its isometry group $\mathrm{SO}(2,2)$ with respect to the isotropy subgroup $\mathrm{SO}(2,1)\subset \mathrm{SO}(2,2)$. 
We start by analysing the quasitriangular Lie bialgebra structure on $\mathfrak{so}(2,2)$.  We show that  there are many  Poisson-Lie structures on $\mathfrak{so}(2,2)$ for which $\mathfrak{so}(2,1)$ defines a {coisotropic} Lagrangian Lie subalgebra of the double $D(\mathfrak{so}(2,2))$ but only a
three-parameter family of classical $r$-matrices that gives $\mathrm{SO}(2,1)$ the structure of a Poisson subgroup of $\mathrm{SO}(2,2)$. In Sections 4.2 and 4.3 we then construct the coisotropic Poisson homogeneous structures on $\mathrm{AdS}_3$ that are associated with the two classical double structures $\mathfrak{so}(2,2)=D(\mathfrak a)$. In both cases, we give an explicit parametrisation of these Poisson structures in terms of coordinates and analyse their 
limit of vanishing cosmological constant  as well as their linearisation. Section 5 contains concluding remarks and closes the paper.


\section{Poisson homogeneous spaces}

In this section we summarise the basics on  Poisson-Lie groups and Poisson homogeneous spaces  and recall the result by Drinfel'd \cite{DrHS}, which states that Poisson homogeneous structures on the quotient $G/H$ of a Poisson-Lie group $G$ with respect to a Lie subgroup $H\subset G$ correspond to the orbits of a natural $G$-action  on the  variety of Lagrangian
subalgebras of the  classical double $D(\gothg)$. Basic references on  Poisson homogeneous spaces are~\cite{DrHS, Zakrzewski, Reyman, Karolinski, KS, CGDoc}. 

\subsection{Poisson homogeneous spaces over $G$  and Lagrangian Lie subalgebras of $D(\gothg)$}

Recall  that a {\em homogeneous space} over a Lie group $G$ is a smooth manifold $M$ together with a transitive smooth group action $\rhd: G\times M\to M$. A {\em homomorphism of homogeneous spaces} is a smooth map $\phi: M\to N$  that intertwines the $G$-actions on $M$ and $N$, i.~e.~that satisfies   $g\rhd \phi(m)=\phi(g\rhd m)$ for  $g\in G$, $m\in M$.  If $\phi$ is a diffeomorphism, then it is called an {\em isomorphism of homogeneous spaces}.

Every Lie subgroup $H\subset G$ gives rise to a homogeneous space $G/H$ with the 
 group action $\rhd: G\times G/H\to G/H$,  $g\rhd (uH)=(gu)H$.  
 If  $H'=gHg^\inv$ 
 , then $\phi_{g}: G/H\to G/H'$, $uH\mapsto (ug^\inv) H'$ is an isomorphism of homogeneous spaces.
 
Conversely, if $M$ is a homogeneous space over $G$, then  the stabiliser $H_m=\{g\in G:\, g\rhd m=m\}$ of each point  $m\in M$  is a Lie subgroup of $G$. As the $G$-action on $M$  is transitive,   the  map $\rhd_m: G\to M$, $g\mapsto g\rhd m$  induces an isomorphism of homogeneous spaces 
$\phi_m: G/H_m\to M$, $uH_m\mapsto u\rhd m$. For any   $m'\in M$, there is a $g\in G$ with 
 $m'=g\rhd m$. This implies  $H_{m'}=gH_mg^\inv$ and that the 
 map $\phi_g: M\to M$, $m\mapsto g\rhd m$ is an isomorphism of homogeneous spaces. 
 
The concept of a {\em Poisson} homogeneous space over a {\em Poisson-Lie group} $G$ is obtained by requiring that all manifolds in the definition of a homogeneous space are Poisson manifolds and all smooth maps Poisson maps.

\begin{definition} \label{def:poissohom}Let $G$ be a Poisson-Lie group. 
\begin{enumerate}
\item A {\bf Poisson homogeneous space} over $G$ is a Poisson manifold $M$ with a transitive group action $\rhd: G\times M\to M$, $(g,m)\mapsto g\rhd m$ that is a Poisson map with respect to the Poisson structure on $M$ and the product of the Poisson structures on $G$ and $M$.

\item  A {\bf homomorphism of Poisson homogeneous spaces} over $G$ is a Poisson map $\phi: M\to N$ that intertwines the $G$-actions on $M$ and $N$: $g\rhd\phi(m)=\phi(g\rhd m)$ for all $g\in G$ and $m\in M$. 
\item An {\bf isomorphism of Poisson homogeneous spaces} is a homomorphism of Poisson homogeneous spaces that  is a diffeomorphism. 
\end{enumerate}
\end{definition}

Note that by definition every Poisson-Lie group $G$ is a Poisson homogeneous space over itself with respect to its action on itself by left multiplication. 
Explicitly, the condition that the action $\rhd: G\times M\to M$ is a Poisson map reads
\begin{align}\label{eq:pmap}
\{f,h\}(g\rhd m)=\{f\circ \rhd_g, h\circ \rhd_g\}(m)+\{f\circ \rhd_m, h\circ\rhd_m\}_G(g),
\end{align}
where $\{\,,\,\}$ denotes the Poisson bracket on $M$,  $\{\,,\,\}_G$  the Poisson bracket on $G$ and \begin{align}
\rhd_m: G\to M, \quad g\mapsto g\rhd m \qquad  \rhd_g: M\to M, \quad m\mapsto g\rhd m.
\end{align}
Poisson homogeneous spaces were classified in \cite{DrHS} via the associated infinitesimal structures, namely Lagrangian Lie subalgebras of the Drinfeld double $D(\gothg)$. The notion of a Lagrangian Lie subalgebra is obtained from the fact that  $D(\gothg)=\gothg\oplus\gothg^*$ as a vector space and the pairing between $\gothg$ and $\gothg^*$.

\begin{definition}\label{def:lagrangian} A {\bf Lagrangian Lie subalgebra} of the Drinfeld double $D(\gothg)$ is a Lie subalgebra $\mathfrak l\subset D(\gothg)$ that satisfies $\mathfrak l^\bot=\mathfrak l$ with respect to the pairing 
$$
\langle (X,\alpha), (Y,\beta)
\rangle=\alpha(Y)+\beta(X)\qquad\forall X,Y\in \gothg,\alpha,\beta\in \gothg^*$$
on $D(\gothg)$. A Lagrangian Lie subalgebra $\mathfrak l\subset D(\gothg)$ is called {\bf coisotropic} if 
$$(\mathfrak l\cap \gothg)^\bot:=\{\alpha\in \gothg^*:\alpha(X)=0\;\forall X\in \gothg\}=\mathfrak l\cap\gothg^*.$$
\end{definition}

\begin{theorem}\cite{DrHS} Let $G$ be a Poisson-Lie group.
\begin{enumerate}
\item Every pair $(M,m)$ of a Poisson homogeneous space $M$  over $G$ and a point $m\in M$ defines a Lagrangian Lie subalgebra $\mathfrak l\subset D(G)$ with $\mathfrak l\cap \gothg=\mathfrak h_m$, where $\gothh_m$ is the Lie algebra of the stabiliser of $m$.

\item Isomorphism classes of Poisson homogeneous spaces over $G$ correspond bijectively to $G$-orbits
of pairs $(\mathfrak l, H)$, where $\mathfrak l\subset D(\mathfrak g)$ is a Lagrangian Lie subalgebra and $H\subset G$ a Lie subgroup with $\text{Lie}\, H=\mathfrak h=\mathfrak l\cap \mathfrak g$.
 \end{enumerate}
\end{theorem}

  The proof that every Lagrangian Lie subalgebra $\mathfrak l\subset D(\gothg)$ with $\mathfrak l\cap \gothg=\gothh$ gives rise to a Poisson homogeneous structure on $G/H$ is essentially obtained by exponentiation, from uniqueness results on Poisson-Lie groups and  the condition \eqref{eq:pmap}, that relates the Poisson brackets at different points in $G/H$ to each other and the Poisson bracket on $G$.

That every pair $(M,m)$ of a Poisson homogeneous space $M$ over $G$ and a point $m\in M$ determines   a Lagrangian Lie subalgebra $\mathfrak l\subset D(\gothg)$ with $\mathfrak l\cap \gothg=\gothh$, where $\gothh$ is the stabiliser of $m$ can be seen as follows. Denoting by $H$ the stabiliser of $m\in M$, one obtains a diffeomorphism $\phi:  G/H\to M$, $gH\mapsto g\rhd m$.  
This identifies  $T_mM\cong T_{eH}G/H\cong \gothg/\gothh$  and $T^*_mM\cong T^*_{eH} G/H=\gothh^\bot=\{\alpha\in \gothg^*:\alpha(X)=0\;\forall X\in \gothh\}$. One then considers the  linear map $T^*_{eH}G/H\cong \gothh^\bot\to \gothg/\gothh\cong T_{eH}G/H$ induced by the Poisson bivector on $M\cong G/H$.
As the  proof from \cite{DrHS}, see also \cite{Reyman} for a pedagogical  summary,  works with canonical structures,  while we require a choice of basis,
we summarise this implication in a more explicit way. 

For every Poisson homogeneous space $M$ over $G$  and $m\in M$
the infinitesimal action $\phi_m:\gothg\to TM$, $X\mapsto X^\rhd(m)$ that assigns to  $X\in \gothg$ the vector 
$X^\rhd(m)\in T_mM$  generated by the $G$-action 
\begin{align}\label{eq:actdefvec}
(X^\rhd. f)(m)=\frac d {dt}\vert_{t=0} f(e^{-tX}\rhd m)\qquad\forall X\in \gothg, m\in M, f\in C^\infty(M),
\end{align}
is a Lie algebra homomorphism with kernel $\ker(\phi_m)=\gothh_m=\mathrm{Lie}\, H_m$. 
The vector fields $X^\rhd$
are related to the right invariant vector fields $X^L$ on $G$ by
\begin{align}\label{eq:actvf}
X^L.(f\circ\rhd_m)(e)=(X^\rhd.f)(m)\qquad\qquad\forall X\in\gothg, m\in M, f\in C^\infty(M),
\end{align} and form a Lie subalgebra $\mathrm{Vec}^\rhd(M)\subset \mathrm{Vec}(M)$.
As the action of $G$ on $M$ is transitive, one has $T_mM=\mathrm{Span}_\R \{X^\rhd(m)\,|\, X\in \gothg\}\cong \mathfrak g/\gothh_m$ and $T^*_mM=(\gothg/\gothh_m)^*\cong \gothh_m^\bot$ for all $m\in M$. 
The Poisson bivector $\pi\in TM\wedge TM$ defines  for every point $m\in M$  a linear map 
\begin{align}\label{eq:bivmap}
\Pi_m:T^*_mM\to T_mM, \quad \alpha\mapsto (\alpha\oo\id)(\pi(m)).
\end{align}
 Identifying $T_mM=\gothg/\gothh_m$ and $T^*_m M=\gothh_m^\bot$ and denoting by $p_m: \gothg\to \gothg/\gothh_m$ the canonical surjection,  one can identify the graph of this map with the linear subspace
\begin{align}\label{eq:lgrdef}
\mathfrak l_m=\{(X,\alpha)\in \gothg\oplus\gothh_m^\bot\,:\, \Pi_m(\alpha)=p_m(X) \}\subset \gothg\oplus\gothg^*=D(\gothg).
\end{align}
The proof that every Poisson homogeneous structure on $M=G/H$ defines a Lagrangian Lie subalgebra of $D(\gothg)$ is then obtained by proving the following two statements:
\begin{enumerate}
\item The antisymmetry of the Poisson bivector $\pi$ guarantees that $\mathfrak l:=\mathfrak l_{eH}\subset D(\gothg)$ is a Lagrangian  subspace with respect to the pairing in Definition \ref{def:lagrangian}.
\item The Jacobi identity of the Poisson bracket on $M$ guarantees that $\mathfrak l:=\mathfrak l_{eH}\subset D(\gothg)$ is a Lie subalgebra.
\end{enumerate}
Although these statements can be proven in a more elegant way, it is also possible to verify them by direct computations in terms of a basis.  
If  $M=G/H$ for a Lie subgroup $H\subset G$ and $m=eH$, one has $\gothh_m=\gothh$,  $T_mM\cong \gothg/\gothh$ and $T^*_mM\cong \gothh^\bot$.
We choose a basis $\{H_1,...,H_n\}$ of $\gothh$ and complete it to a basis $\{H_1,...,H_n, T_{n+1},...,T_N\}$  of $\gothg$. The dual basis of $\gothg^*$  then takes the form  $\{h^1,...,h^n, t^{n+1},...,t^N\}$ with  the pairing given by
\begin{align}\label{eq:pair}
&\langle H_i,h^j\rangle=\langle h^j, H_i\rangle=\delta_i^j & &\langle T_\alpha, t^\beta\rangle=\langle t^\beta, T_\alpha\rangle=\delta_\alpha^\beta\\
&\langle H_i,t^\alpha\rangle=\langle t^\alpha, H_i\rangle=0 & &\langle T_\alpha, h^i\rangle=\langle h^i,T_\alpha\rangle=0,\nonumber
\end{align}
for all $i,j\in\{1,..,n\}$ and $\alpha,\beta\in\{n+1,...,N\}$.  For any basis
 $\{X_i\}$  of $\gothg$, in which the Lie bracket and cocommutator of $\gothg$ are given by structure constants $C_{ij}^k$ and $f_{i}^{jk}$
\begin{align}
[X_i,X_j]=C_{ij}^k X_k\,,\qquad\delta(X_i)=f_i^{jk} X_j\oo X_k\, ,
\end{align}
 the Lie algebra structure of  $D(\gothg)$  is given by
\begin{align}
\label{doublelb}
[X_i,X_j]&=C_{ij}^k X_k\,, & [x^i,x^j]&=f^{ij}_k x^k\,, & [x^i,X_j]&=C^i_{jk}x^k-f^{ik}_j X_k\,,
\end{align}
where $\{x^i\}$  is the dual basis of $\gothg^*$. For the basis $\{H_1,...,H_n, T_{n+1},...,T_N\}$ this takes the form
\begin{align}
[H_i,H_j]&=C_{ij}^k H_k+C_{ij}^{\al}T_{\al} & [h^i,h^j]&=f^{ij}_k h^k+f^{ij}_{\al}t^{\al} \nonumber \\
[H_i,T_{\al}]&=C_{i\al}^j H_j+C_{i\al}^{\bb}T_{\bb} & [h^i,t^{\al}]&=f^{i\al}_j h^j+f^{i\al}_{\bb}t^{\bb} \nonumber \\
[T_{\al},T_{\bb}]&=C_{\al\bb}^{\ga}T_{\ga}+C_{\al\bb}^i H_i & [t^{\al},t^{\bb}]&=f^{\al\bb}_{\ga}t^{\ga}+f^{\al\bb}_ih^i \label{doublesplit} \\
[h^i,H_j]&=C^i_{jk}h^k+C^i_{j\al}t^{\al}-f^{ik}_jH_k-f^{i\al}_jT_{\al} & [h^i,T_{\al}]&=C^i_{\al j}h^j+C^i_{\al\bb}t^{\bb}-f^{ij}_{\al}H_j-f^{i\bb}_{\al}T_{\bb}     \nonumber \\
[t^{\al},H_i]&=C^{\al}_{ij}h^j+C^{\al}_{i\bb}t^{\bb}-f^{\al j}_iH_j-f^{\al\bb}_iT_{\bb} & [t^{\al},T_{\bb}]&=C^{\al}_{\bb i}h^i+C^{\al}_{\bb\ga}t^{\ga}-f^{\al i}_{\bb}H_i-f^{\al\ga}_{\bb}T_{\ga}.  \nonumber 
\end{align}
where Latin indices run over $\{1,...,n\}$ and Greek indices over $ \{n+1,...,N\}$. The condition that $H\subset G$ is a Lie subgroup implies $C_{ij}^\alpha=0$ for all $i,j\in \{1,...,n\}$ and $\alpha\in \{n+1,...,N\}$.

By definition of the basis,  there is a neighbourhood $U\subset G/H$  of $eH\in G/H$, where the vectors $T^\rhd_\alpha(u)$ for $\alpha\in \{n+1,...,N\}$ form a basis of $T_u U$ for all $u\in U$. Hence,   the Poisson bivector on $U$ can be expressed  as $\pi=\pi^{\alpha\beta} T^\rhd_\alpha\oo T^\rhd_\beta$, with coefficient functions $\pi^{\alpha\beta}\in C^\infty(U)$. 
The  linear subspace $\mathfrak l=\mathfrak l_{eH}$ from \eqref{eq:lgrdef}  is then given by
\begin{align}\label{eq:lexpl}
\mathfrak l=\mathfrak l_{eH}=\gothh\oplus \mathrm{span}_\R\{ t^\alpha+\pi^{\alpha\beta} T_\beta\,:\alpha\in \{n+1,...,N\}\},
\end{align}
where we abbreviate $\pi^{\alpha\beta}:=\pi^{\alpha\beta}(eH)$.
A short computation using the pairing in \eqref{eq:pair} then shows that $l_{eH}$ is a Lagrangian subspace if and only if $\pi^{\alpha\beta}=-\pi^{\beta\alpha}$. Conversely, any Lagrangian subspace $\mathfrak l\subset D(\gothg)$ with $\mathfrak l\cap \gothg=\gothh$ can be brought into the form \eqref{eq:lexpl} with $\pi^{\alpha\beta}=-\pi^{\beta\alpha}$ by applying an invertible linear transformation to the subspace $\mathfrak t=\mathrm{span}_\R\{T_\alpha\}\subset\gothg$ and  $\mathfrak h^\bot=\mathrm{span}_\R\{t^\alpha\}$.

 To show that $\mathfrak l_{eH}\subset D(\gothg)$ is a Lie subalgebra, one computes the Jacobi identity for the Poisson bracket on $U\subset M$ and obtains  for $f_1,f_2,f_3\in C^\infty(U)$
\begin{align}\label{eq:poiss}
0&=\{\{f_1,f_2\}, f_3\}+\{\{f_3,f_1\}, f_2\}+\{\{f_2,f_3\}, f_1\}\\
&=(T^\rhd_\alpha. f_1)(T^\rhd_\beta. f_2)(T^\rhd_\gamma. f_3)
\left(\pi^{\delta\gamma} M^{\alpha\beta}_\delta+\pi^{\delta\alpha} M^{\beta\gamma}_\delta+\pi^{\delta\beta} M^{\gamma\alpha}_\delta+\pi^{\epsilon\beta}\pi^{\delta\gamma}C_{\delta\epsilon}^\alpha+\pi^{\epsilon\alpha}\pi^{\delta\beta}C_{\delta\epsilon}^\gamma+\pi^{\epsilon\gamma}\pi^{\delta\alpha}C_{\delta\epsilon}^\beta\right)\nonumber,
\end{align}
where $\pi^{\alpha\beta}:=\pi^{\alpha\beta}(eH)$ are the coefficient functions of the Poisson bivector, $C_{\alpha\beta}^\delta$ the structure constants of $\gothg$ and 
$
M^{\alpha\beta}_\gamma:=(T^\rhd_\gamma.\pi^{\alpha\beta})(eH)
$. By differentiating condition \eqref{eq:pmap} for the Poisson bracket on $G/H$  in $m=eH$, using the relation between the vector fields $T^\rhd_\alpha$ and the associated right invariant vector fields on $G$ one obtains an expression for these constants in terms of the coefficients of the Poisson bivector and the structure constants of $\gothg$ and $\gothg^*$. The corresponding  expressions for the vector fields $H_i^\rhd$  vanish, since $H\subset G$ is the stabiliser of $eH\in G/H$
\begin{align}\label{eq:coeffs}
&M^{\alpha\beta}_\gamma:=(T^\rhd_\gamma.\pi^{\alpha\beta})(eH) =f^{\alpha\beta}_\gamma+\pi^{\delta\beta}   C_{\gamma\delta}^\alpha+\pi^{\delta\alpha}   C_{\gamma\delta}^\beta\\ 
&M^{\alpha\beta}_i:=(H^\rhd_i.\pi^{\alpha\beta})(eH)=f_i^{\alpha\beta}+\pi^{\delta\beta}   C_{i\delta}^\alpha+\pi^{\delta\alpha}   C_{i\delta}^\beta =0.\nonumber
\end{align}
Using the expressions for the Lie bracket in \eqref{doublesplit} one finds that the Lie brackets of the basis elements of  $\mathfrak l_{eH}$ in \eqref{eq:lexpl} are given by $[H_i,H_j]=C_{ij}^k H_k$ and
\begin{align}\label{eq:liebrell}
[t^\alpha+\pi^{\alpha\gamma}T_\gamma,t^\beta+\pi^{\beta\delta}T_\delta]&=(f^{\alpha\beta}_\epsilon+\pi^{\beta\delta}C^\alpha_{\delta\epsilon}-\pi^{\alpha\delta}C^\beta_{\delta\epsilon})t^\epsilon +(\pi^{\alpha\gamma}\pi^{\beta\delta}C_{\gamma\delta}^\epsilon+\pi^{\alpha\gamma}F_\gamma^{\beta\epsilon}-\pi^{\beta\gamma}f_\gamma^{\alpha\epsilon})T_\epsilon\nonumber\\
&+(f_i^{\alpha\beta}+\pi^{\beta\delta}C^\alpha_{\delta_i}-\pi^{\alpha\delta}C^\beta_{\delta i})h^i+(\pi^{\alpha\gamma}\pi^{\beta\delta}C^i_{\gamma\delta}+\pi^{\alpha\delta}f^{\beta i}_\delta-\pi^{\beta\delta}f^{\alpha i}_\delta)H_i\nonumber\\
[t^\alpha+\pi^{\alpha\gamma}T_\gamma, H_i]&= C^\alpha_{i\beta}t^\beta
-(f^{\alpha\beta}_i-\pi^{\alpha\gamma}C_{\gamma i}^\beta)T_\beta-(f^{\alpha j}_i-\pi^{\alpha\gamma}C_{\gamma i}^j)H_j .
\end{align}
Inserting the condition on the coefficients arising from the Jacobi identity in \eqref{eq:poiss} and the two conditions in \eqref{eq:coeffs} in the last two equations, one obtains
\begin{align}\label{eq:liebrell2}
[t^\alpha+\pi^{\alpha\gamma}T_\gamma,t^\beta+\pi^{\beta\delta}T_\delta]
&=(f^{\alpha\beta}_\epsilon+\pi^{\beta\delta}C_{\delta\epsilon}^\alpha-\pi^{\alpha\delta}C_{\delta\epsilon}^\beta)(t^\epsilon+\pi^{\epsilon\delta}T_\delta)+(\pi^{\alpha\gamma}\pi^{\beta\delta}C^i_{\gamma\delta}+\pi^{\alpha\delta}f^{\beta i}_\delta-\pi^{\beta\delta}f^{\alpha i}_\delta)H_i\nonumber\\
[t^\alpha+\pi^{\alpha\gamma}T_\gamma, H_i] &=C_{i\beta}^\alpha(t^\beta+\pi^{\beta\gamma}T_\gamma)-(f_i^{\alpha j}-\pi^{\alpha\gamma}C_{\gamma i}^{j})H_j,
\end{align}
and hence $\mathfrak l_{eH}$ is a Lie subalgebra of $D(\gothg)$. Conversely, if the brackets \eqref{eq:liebrell} define a Lie algebra structure on $\mathfrak l_{eH}$, one finds that the conditions \eqref{eq:poiss} and the two conditions in \eqref{eq:coeffs} hold.  
Thus, finds that the Jacobi identity on $G/H$ is equivalent to the condition that the linear subspace $\mathfrak l_{eH}\subset D(\gothg)$ is a Lie subalgebra.

A change of the chosen point $m\in M$ leads to a conjugate Lie subgroup $H\subset G$. If $m'=g\rhd m$ with $g\in G$, then the stabiliser of $m'$ is related to the stabiliser of $m$ by $H_{m'}=g H_m g^\inv$.  The transformation of the associated Lagrangian subalgebra  $\mathfrak l_{m}$ is given by $\mathfrak l_{m'}=g\rhd_{D(\gothg)} \mathfrak l_m$, where $\rhd_{D(\gothg)}: G\times D(\gothg)\to D(\gothg)$ is the action of $G$ on $D(\gothg)$ given by
\begin{align}\label{actg}
g\rhd_{D(\gothg)}(X,\alpha)=\Ad_g(X)+\Ad_g^*(\alpha)+(\Ad_g^*(\alpha)\otimes\id)\,\pi_G(g)
\end{align}
for $X\in \gothg$, $\alpha\in \gothg^*$, where $\Ad_g: \gothg\to\gothg$ denotes the adjoint action of $G$ on $\gothg$, $\Ad^*_g\equiv(\Ad_g^{-1})^\ast:\gothg^*\to \gothg^*$ the coadjoint action of $G$ on $\gothg^*$ and $\pi_G: G\to \gothg\oo\gothg$ the Poisson bivector of $G$. As 
this action preserves the pairing, it sends Lagrangian Lie subalgebras of $D(\gothg)$ to Lagrangian Lie subalgebras and hence induces an action of $G$ on the algebraic variety of Lagrangian Lie subalgebras of $D(\gothg)$. 
Hence, isomorphism classes of Poisson homogeneous structures on $G/H$ correspond to  orbits of the $G$-action \eqref{actg} on the algebraic variety of Lagrangian
subalgebras $D(\gothg)$.

While Drinfeld's description \cite{DrHS} of Poisson homogeneous spaces over a Poisson-Lie group $G$ in terms of Lagrangian Lie subalgebras of the classical double $D(\gothg)$ relates the two structures, it does not directly lead to a classification. Deriving an explicit classification result is rather difficult, since it amounts to a classification of the orbits of the $G$-action \eqref{actg} that correspond to Lagrangian subalgebras. For the  case of  complex reductive connected algebraic Poisson-Lie groups $G$  and connected isotropy subgroups this was achieved in \cite{karo}.

\bigskip
In the following, we focus on the coisotropic Lagrangian Lie subalgebras $\mathfrak l\subset D(\gothg)$, for which $\mathfrak l\cap \gothg^*=\gothh^\bot$.  By comparing this condition with the expression  \eqref{eq:lexpl} for the Lagrangian subalgebra in terms of the basis $\{H_1,...,Hn,T_{n+1},...T_N \}$ of $\gothg$ and its dual, one obtains the following  description.

\begin{lemma}\label{lem:coiso} Let $\gothh\subset \gothg$ a Lie subalgebra and  $\mathfrak l\subset D(\gothg)$ a Lagrangian Lie subalgebra with $\mathfrak l\cap\gothg=\gothh$ given as in \eqref{eq:lexpl}. Then $\mathfrak l$ is coisotropic if and only if  $\pi^{\alpha\beta}=0$ for all $\alpha,\beta\in \{n+1,...,N\}$. In this case, one has $\mathfrak l=\gothh\oplus \gothh^\bot=\mathrm{span}_{\mathbb R}\{H_i,t^\alpha\}$ and the Lie bracket on $\mathfrak l$ is given by
\begin{align}\label{coiscr}
&[t^\alpha, t^\beta]=f^{\alpha\beta}_\gamma t^\gamma &
&[t^\alpha,H_i]=C^\alpha_{i\beta} t^\beta+f_i^{j\alpha} H_j &
&[H_i,H_j]=C_{ij}^k H_k
\end{align}
and the structure constants of $\gothg$ and $\gothg^*$ satisfy
$C_{ij}^\alpha=0$, $f^{\alpha\beta}_i=0$ for all $i,j\in\{1,...,n\}$ and $\alpha,\beta\in\{n+1,...,N\}$. 
\end{lemma}

\begin{proof} That $\mathfrak l$ is coisotropic, i.~e.~satisfies $(\mathfrak l\cap \gothg)^\bot=\mathfrak l\cap \mathfrak g^*$ if and only if $\pi^{\alpha\beta}=0$ for all $\alpha,\beta\in\{n+1,...,N\}$ folllows directly from expression \eqref{eq:lexpl} for the Lagrangian subalgebra $\mathfrak l$ and the identities  $\mathfrak l\cap \gothg=\gothh$, $\gothh^\bot=\mathrm{span}_{\mathbb R}\{t^\alpha\}$. The second condition in \eqref{eq:coeffs} then implies $f_i^{\alpha\beta}=0$ for $i\in\{1,...,n\}$ and $\alpha,\beta\in\{n+1,...,N\}$. Setting $f_i^{\alpha\beta}=0$ and $\pi^{\alpha\beta}=0$ in the expressions \eqref{eq:liebrell2} for the Lie bracket on $\mathfrak l$, then yields the expressions in \eqref{coiscr}.
\end{proof}

Note that the bracket \eqref{coiscr} implies that  in the coisotropic case not only $\gothh\subset \gothg$ but also $\gothh^\bot\subset \gothg^*$ is a Lie subalgebra. Coisotropic Lagrangian Lie subalgebras of $D(\gothg)$ can be distinguished further by considering the cocommutator of $\gothh=\mathfrak l \cap\gothg$.
In terms of a basis $\{H_1,...,H_n, T_{n+1},...,T_N\}$ of $\gothg$ as above  the cocommutator of $\gothg$ is given by the Lie bracket of $\gothg^*$ in the right column of \eqref{doublesplit}:
\begin{align}
\delta(H_i)&=f^{jk}_i \,H_j\wedge H_k + f^{j\bb}_i \,H_j\wedge T_\bb + f^{\bb\ga}_i \,T_\bb\wedge T_\ga ,\\
\delta(T_\al)&=f^{jk}_\al \,H_j\wedge H_k + f^{j\bb}_\al \,H_j\wedge T_\bb + f^{\bb\ga}_\al \,T_\bb\wedge T_\ga .\nonumber
\end{align}
If  $\mathfrak l$ is coisotropic, then by Lemma \ref{lem:coiso} the cocommutator on $\gothh=\mathfrak l\cap\gothg$ 
takes the form
\begin{align}
\delta(H_i)&=f^{jk}_i \,H_j\wedge H_k + f^{j\bb}_i \,H_j\wedge T_\bb,
\label{cois}
\end{align}
and hence $\gothh\subset\gothg$ is a sub-Lie bialgebra  and $H\subset G$ a Poisson-Lie subgroup if and only if $f^{j\bb}_i=0$ for all $i,j\in\{1,...,n\}$ and $\beta\in \{n+1,...,N\}$ or, equivalently, $\delta(\gothh)\subset \gothh\wedge\gothh$. This condition is stronger than the coisotropy condition (\ref{cois}). If it holds, then the Lie bracket \eqref{coiscr}  takes the form
\begin{align}
[H_i,H_j]=C_{ij}^k H_k,
\qquad
[t^{\al},t^{\bb}]=f^{\al\bb}_{\ga}t^{\ga},
\qquad
[t^{\al},H_i]=C^{\al}_{i\bb}t^{\bb},
\end{align}
and hence the Lagrangian subalgebra $\mathfrak l$  is  semidirect product 
$\mathfrak l=\gothh\ltimes\gothh^\bot$, where the action of  $\gothh$ on $\gothh^\bot$ is given by the structure constants of
$\gothg$. Hence, among the coisotropic Lagrangian Lie subalgebras, the  ones for which $\gothh\subset \gothg$ is a sub-Lie bialgebra  have a particularly simple form.

\subsection{{Coisotropic} Lagrangian subalgebras of the double of a double}
\label{subsec:doubledouble}
A special situation that is relevant in the application to 3d gravity in Section \ref{sec:3dgrav} is the case where the Poisson-Lie group $G$ itself is the double  of another Poisson-Lie group $A$ and the Poisson homogeneous space is a quotient of the form $M=D(A)/A$.  It 
follows from the classification of 3d Lie bialgebras and 6d classical doubles in~\cite{gomez,Snobl} that the isometry groups $G_\Lambda$ of constant curvature spacetimes in 3d gravity, which are the groups $\mathrm{PSL}(2,\mathbb C)$, $\mathrm{SO}(2,2)$ and $\mathrm{ISO}(2,1)$ for, respectively, $\Lambda>0$, $\Lambda<0$ and $\Lambda=0$ can all be realised as doubles of certain Poisson-Lie structures on $\mathrm{SL}(2,\R)$. It was shown in~\cite{BHMcqg} that this can be described in a unified framework that involves the cosmological constant $\Lambda$ as a deformation parameter and that the homogeneous constant curvature  spacetimes are given as quotients of the doubles by the group $\mathrm{SL}(2,\mathbb R)$. This provides a strong motivation to investigate Poisson homogeneous spaces of the type $M=D(A)/A$ and the associated Lie bialgebra structures.

For this, 
consider a Lie bialgebra $(\gotha,\delta)$ with a basis $\{X_i\}$ and denote by $x^i$ the dual basis of $\gotha^*$. Then the Lie brackets of the classical double  $D(\gotha)$ are given by
\begin{align}
[X_i,X_j]&=D_{ij}^k X_k\,, & [x^i,x^j]&=g^{ij}_k x^k\,, & [x^i,X_j]&=D^i_{jk}x^k-g^{ik}_j X_k\,,
\label{doublea}
\end{align}
and the  cocommutator of $D(\gotha)$ from the canonical $r$-matrix $r=\sum_i{x^i\otimes X_i}$ takes the form
\begin{align}\label{eq:coc}
\delta_D(X_i)=-g_i^{jk} X_j\oo X_k,\quad \delta_D(x^i)=D_{jk}^i x^j\oo x^k.
\end{align}
Consider  $\gothg^\ast\equiv D(\gotha)^\ast$ with the basis $\{y^i,Y_i\}$  dual to $\{X_i,x_i\}$, i.~e.~given by the pairing
\begin{align}\label{eq:pair}
\langle Y_i, x^j\rangle=\langle y^j,X_i\rangle=\delta_i^j,\quad \langle Y_i,X_j\rangle=\langle y^i,x^j\rangle=0.
\end{align}
Then the  Lie bracket of $\gothg^\ast\equiv D(\gotha)^*$ is given by
\begin{align}
[Y_i,Y_j]=D_{ij}^k Y_k,\quad [y^i,y^j]=-g^{ij}_ky^k,\quad[y^i,Y_j]=0,
\label{dualdoublea}
\end{align}
and the Lie bracket of the double $D(\gothg)\equiv D(D(\gotha))$ is given by~\eqref{doublea}, \eqref{dualdoublea} and the crossed brackets
\begin{align}
[y^i,X_j]=&D_{jk}^i y^k+g_j^{ik} (X_k-Y_k) &
[y^i,x^j]=&g_k^{ij} y^k\\
[Y_i,X_j]=&D_{ij}^k Y_k &
[Y_i,x^j]=&g_i^{jk}Y_k-D_{ik}^j (x^k+y^k).\nonumber
\end{align}
Expressions \eqref{doublea} and \eqref{eq:coc} show that $\mathfrak h\equiv \mathfrak a=\mathrm{Span}_\R\{X_i\}\subset D(\mathfrak a)$ is a sub-Lie bialgebra
of $\mathfrak g\equiv D(\mathfrak a)$. The expressions \eqref{eq:pair} for the pairing between $D(\mathfrak a)$ and $D(\mathfrak a)^*$  and the Lie brackets  \eqref{dualdoublea} 
show that 
$\mathfrak a^\bot=\mathrm{Span}_\R\{Y_i\}$ is a Lie subalgebra of $D(\mathfrak a)^*$. This shows that $\mathfrak l=\mathfrak a\oplus \mathfrak a^\bot$ is a {coisotropic} Lagrangian Lie subalgebra of $D(D(\mathfrak a))$ such that $\mathfrak a=\mathfrak l\cap  D(\mathfrak a)$ is a sub-Lie {\em bialgebra} of $D(\mathfrak a)$. Note also that in this case
the Lie bracket of $\mathfrak l$ is given entirely by the structure constants of $\gotha$
\be
[X_i,X_j]=D_{ij}^k X_k, \qquad
[Y_i,Y_j]=D_{ij}^k Y_k, \qquad
[Y_i,X_j]=D_{ij}^k Y_k.
\ee
More precisely,  the Lie algebra $\mathfrak l$ is a semidirect product $\mathfrak l=\mathfrak a\ltimes_{ad}\mathfrak a$, where $\gotha$ acts on itself via the adjoint action. Hence, for any Poisson-Lie group $G=D(A)$ that is a double of another Poisson-Lie group $A$, the Poisson-Lie group $A\subset G$  is a Poisson-Lie subgroup of $G$ and gives rise to a Poisson homogeneous space $M=D(A)/A$.


\section{{Coisotropic} Poisson homogeneous spaces for $SL(2,\R)\simeq SO(2,1)$}

In this section, we determine  some coisotropic Poisson homogeneous structures on homogeneous spaces $M=\mathrm{SL}(2,\R)/H$ associated with
Poisson-Lie structures on $\mathrm{SL}(2,\R)$ and one-parameter subgroups $H\subset \mathrm{SL}(2,\R)$. 

This is motivated on one hand by their geometrical relevance. The group $\mathrm{SL}(2,\R)$ and its subgroups play an essential role in the description of hyperbolic structures on surfaces and in Teichm\"uller theory. Its homogeneous spaces include two-dimensional hyperbolic space $\mathbb H^2$. On the other hand, the constant curvature spacetimes of 3d gravity, their isometry groups and many associated structures   can be obtained from associated structures for $\mathrm{SL}(2,\R)$ and two-dimensional hyperbolic space by analytic continuation techniques, see for instance \cite{bb,cmsc}.   {Although there may be many more non-coisotropic Poisson homogeneous structures on $\mathrm{SL}(2,\R)/H$,  it can be expected that the coisotropic ones are the simplest to quantise, see \cite{CG} and  also \cite{Podles, NM, VS, Podles9, Sheu, Ciccoliqplanes, Leit, BCGSTqse, CHZ, HMS, BRV, Tomatsu}, and most natural for applications in 3d quantum gravity and noncommutative geometry. It is therefore sensible to focus first on the coisotropic case.}

For this, note first that all  abelian Lie subalgebras of $\mathfrak{sl}(2,\R)$ are one-dimensional  and there are exactly three non-conjugate one-dimensional Lie subalgebras of $\mathrm{sl}(2,\R)$. Hence, up to coverings and isomorphisms, there are exactly three  non-trivial two-dimensional homogeneous spaces over $\mathrm{SL}(2,\R)$.  All of these can be realised as submanifolds of three-dimensional Minkowski space via the identification
$\mathrm{SL}(2,\R)\cong \mathrm{SO}(2,1)$:
\begin{itemize}
\item 2d anti de Sitter space  as the one-sheeted hyperboloid: 
$$\mathrm{AdS}_2=\{x\in \R^3: x_0^2-x_1^2-x_2^2=-1\},$$
\item two copies of 2d hyperbolic space as the two-sheeted hyperboloid: 
$$\mathbb H^2\times \Z_2=\{x\in \R^3: x_0^2-x_1^2-x_2^2=1\},$$
\item the light cone  in 3d Minkowski space:
$$L=\{x\in \R^3: x_0^2-x_1^2-x_2^2=0, {x\neq 0}\}.$$
\end{itemize}
We will show that that these three  homogeneous spaces  correspond one-to-one to the three families of Poisson-Lie structures on $\mathrm{SL}(2,\R)$ in  the following sense: each of the three families of Poisson-Lie structures allows one to realise exactly one of these homogeneous spaces as a {coisotropic} Poisson homogeneous space over a {\em Poisson subgroup} of $\mathrm{SL}(2,\R)$, while the other two are realised   {\em only as coisotropic} Poisson homogeneous spaces.

To construct these Poisson homogeneous spaces over $\mathrm{SL}(2,\R)$, we start by considering the Lie algebra  $\mathfrak g=\mathfrak{sl}(2,\R)\cong \mathfrak{so}(2,1)$.  In the standard basis $\{J_\pm, J_3\}$ its Lie bracket  takes the form
\begin{align}
&[J_3,J_\pm ] =\pm  2J_\pm , \quad[J_+,J_-]=J_3,
\end{align}
and we denote by $\{a_\pm,\chi\}$ the  dual basis of  $\gothg^*$ with the pairing
\begin{align}
\langle J_3,\chi\rangle=\langle J_\pm, a_{\pm}\rangle=1,\qquad \langle J_3,a_\pm\rangle=\langle J_\pm,\chi\rangle=\langle J_\pm,a_\mp\rangle=0.
\end{align}
In the sequel we will also use another basis $\{P_1,P_2, J_{12}\}$ of $\mathfrak{sl}(2,\mathbb R)$, which is adapted to the (1+1)-dimensional Cayley-Klein geometries.  The isometry groups of the (1+1)-dimensional Cayley-Klein geometries are given as a two-parameter family of Lie groups
$G_{(\k_1,\k_2)}$ whose Lie algebras $\g_{(\k_1,\k_2)}$ are spanned by $\{P_1,P_2,J_{12}\}$ with the Lie brackets
\begin{align}\label{eq:ckgen}
\conm{J_{12}}{P_{1}}=P_{2}.\qquad
\conm{J_{12}}{P_{2}}=-\k_2 P_{1},\qquad
\conm{P_{1}}{P_{2}}=\k_1 J_{12}. 
\end{align}
The parameters 
$\k_1$ and $\k_2$ are the constant  Gaussian curvatures of, respectively,  the symmetric homogeneous
space of points $G_{(\k_1,\k_2)}/\langle J_{12}
\rangle$  and the symmetric homogeneous space of lines $G_{(\k_1,\k_2)}/\langle P_{1}
\rangle$, which carry a transitive  $G_{(\k_1,\k_2)}$ action by left multiplication. The properties of the nine (1+1)-dimensional geometries corresponding to different signs of $\kappa_1$ and $\kappa_2$ were investigated in detail in \cite{BHOS2d} and their (2+1)-dimensional counterparts in~\cite{BHOS3d}. The Lie algebra $\mathfrak{sl}(2,\R)$ is  the Cayley-Klein algebra \eqref{eq:ckgen} for $\k_1=1$ and $\k_2=-1$  
with the relation between the bases  $\{P_1,P_2,J_{12}\}$ and $\{J_\pm,J_3\}$ given by
\begin{align}
P_1=\tfrac 1 2 (J_+ - J_-),\qquad  J_{12}=\tfrac 1 2 J_3,\qquad P_2=\tfrac 1 2 (J_++J_-).
\end{align}
We denote by $\{a_1,a_2,\theta\}$ the associated dual basis of $\mathfrak{sl}(2,\R)^*$ with the pairing
\begin{align}
\langle P_1,a_1\rangle=\langle P_2, a_2\rangle=\langle J_{12},\theta\rangle=1,\qquad \langle P_1,a_2\rangle=\langle P_2,a_1\rangle=\langle J_{12},a_i\rangle=\langle P_i,\theta\rangle=0.
\end{align}
The Lie algebra $\mathfrak{sl}(2,\R)$ has three non-conjugate one-dimensional Lie subalgebras $\gothh$, that exponentiate to hyperbolic, elliptic and parabolic subgroups $H\subset \mathrm{PSL}(2,\R)$. In terms of the basis $\{J_{12}, P_1,P_2\}$ they are generated, respectively, by
$J_{12}$, by $P_1$ and by $P_1+P_2$. The associated  homogeneous spaces $\mathrm{SL}(2,\R)/H$   are the one-sheeted hyperboloid $\mathrm{AdS}_2$, the two-sheeted hyperboloid $\mathbb H_2\times \Z_2$ and the lightcone $L$ in Minkowski space.

On the other hand, there are exactly three inequivalent families of Lie bialgebra structures on $\mathfrak{sl}(2,\mathbb R)\simeq \mathfrak{so}(2,1)$, that yield three distinct families of Poisson-Lie structures on $SL(2,\R)$, see for instance~\cite{Reyman}. All three of them are quasitriangular  and given by the following families of antisymmetric $r$-matrices:
\begin{itemize}
\item The standard Drinfel'd-Jimbo Lie bialgebra structure, called   {\em hyperbolic} in~\cite{Reyman}, 
\be
r=\eta\,J_+\wedge J_-=2\,\eta\,\,P_1\wedge P_2\qquad\text{for}\;\eta\in\R,
\label{standard}
\ee  
\item The other standard family of Lie bialgebra structures,  called {\em elliptic} in~\cite{Reyman},
\be
r={z}\,\left(J_3\wedge (J_++J_-)\right)= 2 z\,J_{12}\wedge P_2\qquad\text{for}\; z\in\R,
\label{elliptic}
\ee
\item The non-standard or triangular Lie bialgebra structure, called  {\em parabolic} in~\cite{Reyman},
\be
r=\tfrac 1 2 \, J_3\wedge J_+ = J_{12}\wedge (P_1 + P_2).
\label{parabollic}
\ee
\end{itemize}

The names {\em hyperbolic}, {\em elliptic} and {\em parabolic} are motivated by the fact that the one-dimensional Lie subalgebras generated by primitive elements  of the associated cocommutators are hyperbolic, elliptic and parabolic. 
We will now determine  for each of these Lie bialgebra structures on $\mathfrak{sl}(2,\R)$ the Lagrangian Lie subalgebras of the classical double $D(\mathfrak{sl}(2,\R))$
and the associated two-dimensional Poisson homogeneous spaces $M$, with an explicit description of the Sklyanin bracket and the induced Poisson  bracket on $M$.


\subsection{The Poisson homogeneous space $\mathrm{AdS}_2$  (hyperbolic, $\gothh=\mathrm{Span}_\R \{J_{12}\}$)}

The cocommutator of  on $\mathfrak{sl}(2,\R)$ for the standard hyperbolic $r$-matrix~\eqref{standard}
is given by
\begin{align}
\delta(J_{12})=0,\qquad
\delta(P_1)=2\,\eta\,P_1\wedge J_{12},\qquad
\delta(P_2)=2\,\eta\,P_2\wedge J_{12}.
\label{deltastandard}
\end{align}
The classical double $D(\mathfrak{sl}(2,\R))$ for this Lie bialgebra structure is  isomorphic to $\mathfrak{so}(2,2)$ as a Lie algebra~\cite{BHMcqg}.
It follows directly from the expressions for the cocommutator in \eqref{deltastandard} that a
one-dimensional Lie subalgebra $\gothh\subset \mathfrak{sl}(2,\R)$ is a sub-Lie {\em bialgebra} if and only if it is spanned by the generator  $J_{12}$. 
In this case, the associated {coisotropic} Lagrangian subalgebra  $\mathfrak l\subset D(\mathfrak{sl}(2,\R))$ is 
$\mathfrak l=\text{Span}_\R\{ J_{12}, a_1,a_2\},
$
with the following Lie bracket obtained from~\eqref{doublelb}:
\begin{align}
[J_{12},a_1]=  - a_2,
\qquad
[J_{12},a_2]=  - a_1,
\qquad
[a_1,a_2]= 0 .
\end{align}
The Poisson-Lie structure on $\mathrm{SL}(2,\R)$ associated with the Lie bialgebra structure \eqref{deltastandard} is given by the
 Sklyanin bracket
\be
\{f,g\}=  r^{ij}(X^L_i .f\, X^L_j. g -
X^R_i. f\, X^R_j. g) , \label{gb}
\ee
where $r=r^{ij} X_i\oo X_j$ is the $r$-matrix~\eqref{standard} with respect to a basis $\{X_i\}$ of $\mathfrak{sl}(2,\R)$ and $X_i^L$, $X_i^R$ denote the right and left invariant vector fields on $G$. A convenient set of coordinates   in which this bracket takes a particularly simple form are the coordinates
 $(\theta,a_1,a_2)$ adapted to the Cayley-Klein geometries defined in \eqref{eq:coordck} in Appendix \ref{sec:appa}. This action of the  right and left invariant vector fields on $\mathrm{SL}(2,\R)$ on these coordinates is given by  \eqref{eq:ckvecs} and \eqref{eq:ckvecs2} in Appendix \ref{sec:appa}. A straightforward computation~\cite{BHOSpl} shows that in terms of these coordinates the Sklyanin bracket \eqref{gb}  takes the form
 \bea
 && \pois{\theta}{a_1}=-2\,\eta\,\frac{\sin a_1}{\cosh a_2},\\
 && \pois{\theta}{a_2}=-2\,\eta\,\tanh a_2,\label{ell1}\\
&& \pois{a_1}{a_2}=2\,\eta\,\left(\frac{1}{\cosh a_2}- \cos a_1 \right).
\label{phsj12}
\eea
The Poisson homogeneous space associated with the Poisson-Lie subgroup $H=\langle J_{12}\rangle\subset \mathrm{SL}(2,\R)$ 
is the one-sheeted hyperboloid 
 $\mathrm{AdS}_2=\mathrm{SL}(2,\R)/\langle J_{12}\rangle$.  It is parametrised by the coordinates $a_1,a_2$, which generate a Poisson subalgebra of $C^\infty(\mathrm{SL}(2,\R))$ with the Poisson bracket given by~\eqref{phsj12}.

An alternative set of coordinate functions on $\mathrm{SL}(2,\R)$ that is adapted to the basis $\{J_\pm,J_3\}$ is  given by the coordinate functions $(a_+,a_-,\chi)$ defined in \eqref{kkb} in Appendix \ref{sec:appa}. The expressions for the left and right invariant vector fields for the basis $\{J_3,J_{\pm}\}$ of $\mathfrak{sl}(2,\R)$ are given in \eqref{kkc} and \eqref{kkc2}. Again, a straightforward computation  shows that the Sklyanin bracket \eqref{gb} for the  $\mathrm{SL}(2,\R)$ with the $r$-matrix~\eqref{standard} takes the form
\bea
&& \{a_+,a_-\}=-2\,\eta\, a_+\,a_-,\label{phsj12pm}\\
&& \{\chi,a_+\}=-\eta\, a_+,\\
&& \{\chi,a_-\}=-\eta\, a_-.\label{par1}
\eea
In these coordinates, the Poisson homogeneous space $\mathrm{AdS}_2=\mathrm{SL}(2,\R)/\langle J_{12}\rangle=\mathrm{SL}(2,\R)/\langle J_{3}\rangle$
is parametrised by the coordinate functions $a_+$ and $a_-$ which generate a Poisson subalgebra of $C^\infty(\mathrm{SL}(2,\R))$ with the Poisson bracket~\eqref{phsj12pm}. 

\subsection{The Poisson homogeneous space $\mathbb H_2\times\Z_2$  (elliptic, $\gothh=\mathrm{Span}_\R \{P_1\}$)}

We can analyse in the same manner the {coisotropic} Poisson homogeneous spaces for the Poisson-Lie structure on $\mathrm{SL}(2,\R)$ given by the 
 elliptic $r$-matrix~\eqref{elliptic}. In this case the cocommutator of $\mathfrak{sl}(2,\R)$ reads
\begin{align}
\delta(J_{12})=2\,z\,J_{12} \wedge P_1,\qquad
\delta(P_1)=0,\qquad
\delta(P_2)=2\,z\,P_2\wedge P_1.
\label{deltatl}
\end{align}
In this case, a one-dimensional Lie subalgebra $\gothh\subset \mathfrak{sl}(2,\R)$ is a sub-Lie {\em bialgebra} if and only if it is of the form  $\gothh=\mathrm{Span}_\R\{ P_1\}$.  The associated Lagrangian subalgebra  of $D(\mathfrak{sl}(2,\R))$ is
$\mathfrak l=\text{Span}_\R\{ P_1, \theta, a_2\}$ with the Lie bracket
\begin{align}
[P_1,\theta]=  - a_2,
\qquad
[P_1,a_2]=  \theta,
\qquad
[\theta,a_2]= 0 .
\end{align}
In terms of the coordinates $(\theta, a_1,a_2)$ from \eqref{eq:coordck} the  Sklyanin bracket on $\mathrm{SL}(2,\R)$ for \eqref{elliptic}
reads
 \bea
 && \pois{\theta}{a_1}=2\,z\,\frac{\sinh \theta}{\cosh a_2},\\
 && \pois{\theta}{a_2}=-2\,z\,\left(\frac{1}{\cosh a_2}- \cosh \theta \right),\label{phstime}\\
&& \pois{a_1}{a_2}=-2\,z\,\tanh a_2.\label{hyp2}
\eea
The associated Poisson homogeneous space with respect to the Poisson subgroup $H=\langle P_1\rangle\subset \mathrm{SL}(2,\R)$ is the two-sheeted hyperboloid
 $\mathbb H^2\times \Z_2=\mathrm{SL}(2,\R)/\langle P_1\rangle$. It is parametrised by the coordinates $a_2,\theta$, which generate a Poisson subalgebra of $C^\infty(\mathrm{SL}(2,\R))$ with the Poisson bracket~\eqref{phstime}.
In terms of the coordinates $(a_+,a_-,\chi)$ from  \eqref{kkb} the Sklyanin bracket on $\mathrm{SL}(2,\R)$ takes the form
\bea
&& \{a_+,a_-\}=-2\,z\,a_-\,(1 + a_+\,a_-)-2\,z\,a_+,\label{hyp2a}\\
&& \{\chi,a_+\}=-z\, (1-e^{2\chi}) +z\,a_+^2\,e^{-2\chi},\\
&& \{\chi,a_-\}=-z(1-e^{-2\chi})- z\,a_-^2,
\label{par2}
\eea
but there is no simple parametrisation of the homogeneous space $\mathbb H^2\times \Z_2=\mathrm{SL}(2,\R)/\langle  J_+-J_-\rangle$  in terms of the coordinates  $(a_+,a_-,\chi)$. Nevertheless,  as we will see  in the next subsection, this bracket gives a simple description of the {\em coisotropic} Poisson homogeneous structures on the lightcone $L$.


\subsection{The Poisson homogeneous space $L$  (parabolic, $\gothh=\mathrm{Span}_\R \{J_+\}$)}

Finally, we analyse the {coisotropic} Poisson homogeneous spaces for the Poisson-Lie structure on $\mathrm{SL}(2,\R)$ given by the 
nonstandard or parabolic  $r$-matrix~\eqref{parabollic}.  In this case the cocommutator of $\mathfrak{sl}(2,\R)$ reads
\begin{align}
\delta(J_{12})=z\,J_{12} \wedge (P_1+ P_2) ,\qquad
\delta(P_1)=z\,P_1\wedge P_2 ,\qquad
\delta(P_2)=z\,P_2\wedge P_1,
\end{align}
and a one-dimensional Lie subalgebra $\gothh\subset\mathfrak{sl}(2,\R)$ is a sub-{\em Lie bialgebra} of $\mathfrak{sl}(2,\R)$ if and only if it is generated by the primitive element $P_1+P_2= J_+$.  In terms of the basis $\{J_\pm, J_3\}$ and its dual  basis $\{a_\pm,\chi\}$, the associated Lagrangian subalgebra  is
$\mathfrak l=\text{Span}_\R\{ J_+, \chi,a_-\}$ with the Lie bracket
\begin{align}
[J_+,\chi]=  - a_-,
\qquad
[J_+,a_-]=  0,
\qquad
[\chi,a_-]=  0.
\end{align}
In terms of the coordinates $(a_+,a_-,\chi)$ from  \eqref{kkb} the Sklyanin bracket on $\mathrm{SL}(2,\R)$ for \eqref{parabollic} reads
\bea
&& \{a_+,a_-\}=-\,a_-\,(1 + a_+\,a_-),\label{hyp3}\\
&& \{\chi,a_+\}=-\tfrac12\, (1-e^{2\chi}),\\
&& \{\chi,a_-\}=-\tfrac 1 2 \,a_-^2. \label{conephs}
\eea
The  Poisson homogeneous space with respect to the Lie subgroup $H=\langle J_+\rangle\subset \mathrm{SL}(2,\R)$ is the lightcone in 3d Minkowski space $L=\mathrm{SL}(2,\R)/\langle J_+\rangle$. It is parametrised by the coordinates $a_-,\chi$, which generate a Poisson subalgebra of $C^\infty(\mathrm{SL}(2,\R))$ with the Poisson bracket~\eqref{conephs}.
 In terms of the  coordinates $(\theta, a_1,a_2)$ from \eqref{eq:coordck} the Sklyanin bracket for \eqref{parabollic} is given by
\bea
 && \pois{\theta}{a_1}= \frac{1}{\cosh a_2} (e^\theta-\cos a_1),\\
 && \pois{\theta}{a_2}=e^\theta-  \frac{1}{\cosh a_2},\label{nonstandardCK}\\
&& \pois{a_1}{a_2}=\sin a_1 - \tanh a_2.\label{nonstandardAdS}
\eea
Again, there is no simple  way to parametrise the Poisson homogeneous space $L=\mathrm{SL}(2,\R)/\langle P_1+P_2\rangle$ in terms of the coordinates $(\theta, a_1,a_2)$. Nevertheless, these brackets  yield a simple description of the {\em coisotropic} 
Poisson homogeneous structures on the two-sheeted hyperboloid.

\subsection{{The full list of coisotropic} Poisson homogeneous spaces}

The discussion in the preceding sections shows that for each of the three inequivalent 2d homogeneous spaces  $M=\mathrm{SL}(2,\R)/H$, there 
is exactly one Poisson-Lie structure on $\mathrm{SL}(2,\R)$ for which the subgroup $H\subset \mathrm{SL}(2,\R)$ is a {\em Poisson-Lie subgroup}.
For the one-sheeted hyperboloid  $M=\mathrm{AdS}_2$ with $H=\langle J_{12}\rangle=\langle J_3\rangle$ it is given by the hyperbolic $r$-matrix \eqref{standard}, for the two-sheeted hyperboloid $\mathbb H_2\times\mathbb Z_2$ with $H=\langle P_1\rangle$ by the elliptic $r$-matrix \eqref{elliptic} and for the lightcone $L$  with $H=\langle J_+\rangle$ by the parabolic $r$-matrix \eqref{parabollic}.   However, each of these homogeneous spaces can be realised as a {\em coisotropic} Poisson homogeneous space over $\mathrm{SL}(2,\R)$ for any of these three Poisson-Lie structures.

The Poisson homogeneous space $M=\mathrm{AdS}_2$ with $H=\langle J_{12}\rangle=\langle J_3\rangle$ has a simple parametrisation in both, the coordinates $(\theta, a_1,a_2)$ from \eqref{eq:coordck} and $(a_+,a_-,\chi)$ from \eqref{kkb}. In terms of the former, it is parametrised by the coordinates
$a_1,a_2$ and in terms of the latter by the coordinates $a_+,a_-$. Expressions \eqref{phsj12}, \eqref{phsj12pm},  \eqref{hyp2}, \eqref{hyp2a}, \eqref{hyp3} and \eqref{nonstandardAdS} show that for all three $r$-matrices, these coordinate functions generate a Poisson subalgebra of $C^\infty(\mathrm{SL}(2,\R))$, but the concrete form of the Poisson brackets depends on  the choice of the classical $r$-matrix.

The Poisson homogeneous space $M=\mathbb H^2\times\mathbb Z_2$ with $H=\langle P_1\rangle$ has a simple parametrisation only in terms of the coordinates  $(\theta, a_1,a_2)$ from \eqref{eq:coordck}, where it is parametrised by  the coordinates $a_2,\theta$. Again, it is apparent in expressions \eqref{ell1}, \eqref{phstime} and \eqref{nonstandardCK} that for all choices of the $r$-matrix, these coordinates generate a Poisson subalgebra of $C^\infty(\mathrm{SL}(2,\R))$. 

In contrast, the Poisson homogeneous space $M=L$ with $H=\langle J_+\rangle$ has a simple parametrisation only in the coordinates $(a_+,a_-,\chi)$ from \eqref{kkb}, where it is parametrised by $a_-,\chi$. The expressions for their Poisson brackets in \eqref{par1}, \eqref{par2} and \eqref{conephs} for the three different $r$-matrices again show that in all three cases, these coordinate functions generate a Poisson subalgebra of $C^\infty(\mathrm{SL}(2,\R))$. 

These results are summarised in  Table 1. The {coisotropic}  Poisson homogeneous structures of the Poisson subgroup type appear in the diagonal, while all off-diagonal PHS {are  coisotropic, but not Poisson subgroup ones}. The first row contains the three {coisotropic} Poisson homogeneous structures for the homogeneous space $\mathrm{AdS}_2$, the first of them being the Poisson subgroup one. Note that this is the only Poisson bracket whose linearisation vanishes. The second  and third rows contain the
{coisotropic}  Poisson homogeneous structures for the homogeneous spaces $\mathbb H^2\times\Z_2$ and for the lightcone $L$, respectively. 
Again, the only {coisotropic}  Poisson homogeneous structures  whose linearisation vanishes are the Poisson subgroup ones. 
This fact already indicates that the Poisson subgroup structures are algebraically simpler, since the first order quantisation of the Poisson homogeneous structures leads to an {\em abelian} quantum spacetime.

\begin{table}[t] \label{tab1}{\scriptsize
\begin{tabular} {|c|c|c|c|}
 \hline
 & & &  \\[-1.5ex]
  & $r=2\,\eta\,P_1\wedge P_2$ &   $r=2\, z\,J_{12}\wedge P_2$ & \!\!\!\!  $r=J_{12}\wedge (P_1 + P_2)=\tfrac 1 2 J_3\wedge J_+$\!\!\!\! 
 \\[+1.5ex]
 \hline
 & & &  \\[+0.5ex]
   $M=SO(2,1)/\langle J_{12}\rangle$ & \!\!\!$\pois{a_1}{a_2}=2\,\eta\,\left(\frac{1}{\cosh a_2}- \cos a_1 \right)$\!\!\! & $ \pois{a_1}{a_2}=-2\,z\,\tanh a_2$ & $\pois{a_1}{a_2}=\sin a_1 - \tanh a_2$ \\[-0.5ex]
(one-sheeted hyperboloid $\mathrm{AdS}_2$) & & &  \\[+1.5ex]
 \hline
 & & &  \\[+0.5ex]
   $M=SO(2,1)/\langle P_1\rangle$ & $ \pois{\theta}{a_2}=-2\,\eta\,\tanh a_2$ & \!\!\!\! $\pois{\theta}{a_2}=-2\,z\,\left(\frac{1}{\cosh a_2}- \cosh \theta \right)$\!\!\!\!  & $\pois{\theta}{a_2}=e^\theta-  \frac{1}{\cosh a_2}$ \\[-0.5ex]
\!\!\!\!(two-sheeted hyperboloid $\mathbb H^2\times \Z_2$)\!\!\!\! & & &  \\[+1.5ex]
 \hline
& & &  \\[+0.5ex]
    $M=SO(2,1)/\langle J_+\rangle$&   $\{\chi,a_-\}=-\eta\, a_-$ & $\{\chi,a_-\}=-z(1-e^{-2\chi})- z\,a_-^2$ & $\{\chi,a_-\}=-\tfrac 1 2\,a_-^2$ \\[-0.5ex]
 (light cone $L$) & & &  \\[+1.5ex]
\hline
 \end{tabular}
\hfill}
\begin{caption}
{2d Poisson homogeneous spaces for the three families of Poisson-Lie structures on $\mathrm{SL}(2,\R)$.}
\end{caption}
\end{table}

It should be noted that the results in this section do not classify all coisotropic Poisson homogeneous structures on $\mathrm{AdS}_2$, $\mathbb H^2\times\mathbb Z^2$ and $\mathrm{L}$, since we consider only those Poisson homogeneous structures on $\mathrm{PSL}(2,\mathbb R)/H$ for which the Lagrangian Lie subalgebra $\mathfrak l\subset D(\mathfrak{sl}(2,\R))$ associated with the point $eH\in \mathrm{SL}(2,\mathbb R)$ is coisotropic.  A full classification of all Poisson homogeneous structures on $\mathbb H^2$ with respect to all Poisson-Lie group structures on $\mathrm{PSL}(2,\mathbb R)$ is given in  \cite{Leit}. It is shown there  that there is a one-parameter 
 family of Poisson homogeneous structures on $\mathbb H^2$, partly of
coisotropic type. This shows  that  on a single homogeneous space $G/H$ there
may be a continuum of Poisson homogeneous structures with respect to a fixed Poisson-Lie  structure on $G$.   
 
 The $\mathrm{PSL}(2,\mathbb R)$-orbits in the variety of Lagrangian Lie-subalgebras of $D(\mathfrak{sl}(2,\mathbb R))$ for the different Lie bialgebra structures on $\mathfrak{sl}(2,\mathbb R)$  are classified in \cite{Ciccoli}. Although it is difficult to transform the results in \cite{Ciccoli} into an explicit description in terms of  coordinates on $\mathrm{AdS}_2$, $\mathbb H^2\times\mathbb Z^2$ and $\mathrm{L}$, this amounts to a complete classification of all Poisson-homogeneous structures on the homogeneous spaces $\mathrm{PSL}(2,\mathbb R)/H$ for all Poisson-Lie structures on $\mathrm{PSL}(2,\mathbb R)$ and Lie subgroups $H\subset \mathrm{PSL}(2,\mathbb R)$.

\section{$\mathrm{AdS}_3$ as a {coisotropic} Poisson homogeneous space over a  double}
\label{sec:cc3dgr}

In this section, we consider the Lorentzian 3d constant curvature spacetimes $X_\Lambda$  as  homogeneous spaces over their isometry groups.  These are 3d anti de Sitter space $\mathrm{AdS}_3$,  de Sitter space $\mathrm{dS}_3$ and  Minkowski space $M_3$ for, respectively, negative, positive and vanishing cosmological constant $\Lambda$. Their isometry groups $G_\Lambda$ are the groups
$\mathrm{SO}(2,2)\cong \mathrm{SL}(2,\R)\times \mathrm{SL}(2,\R)$, $\mathrm{SL}(2,\mathbb C)$ and $\mathrm{ISO}(2,1)$. All of these groups contain the 3d Lorentz group $\mathrm{SL}(2,\R)\cong\mathrm{SO}(2,1)$ as a subgroup, with the diagonal embedding in the first case and the canonical inclusions in the last two cases,  and the spacetimes are given as homogeneous spaces $X_\Lambda=G_\Lambda/\mathrm{SL}(2,\R)$. 

These spacetimes play a distinguished role in 3d gravity since they capture all {\em local} information about 3d spacetimes. Any  Lorentzian vacuum spacetime or spacetime with point particles  is locally isometric to one of these spacetimes and maximally globally hyperbolic spacetimes are obtained as quotients of certain regions in these spacetimes by discrete subgroups of their isometry groups \cite{mess,bb}. 
Moreover, the quantum group symmetries in  quantum gravity arise naturally  from Poisson-Lie symmetries on their isometry groups \cite{FR,AM}, which act on the classical phase space of 3d gravity.
 
Hence it is natural to ask  if it is possible to realise the Lorentzian constant curvature spacetimes in 3d gravity  as  {\em Poisson} homogeneous spaces with respect to a suitable Poisson-Lie structure on their isometry groups $G_\Lambda$.  As shown in \cite{cm2,cmn} the relevant Poisson-Lie structures for 3d gravity are classical doubles.  Hence, we will restrict attention to  quasitriangular  and, in particular,   classical double Poisson-Lie structures in the following.  {As in the previous section, we will focus on the Poisson homogeneous spaces that are {\em coisotropic} and, in particular,  the {\em coisotropic Poisson-subgroup} ones. Although there may be many more Poisson homogeneous structures on  these constant curvature spacetimes, it can be expected that the coisotropic ones are the simplest to quantise.  Also, it turns out  that the coisotropic cases are quite numerous.} We will also focus on the case of negative cosmological constant, i.~e.~3d anti de Sitter space. This is motivated on one hand by its particular geometrical relevance - 3d anti de Sitter spacetimes  exhibit interesting physical properties such as   black hole solutions and conformal structures at their boundaries \cite{carlip}.  On the other hand, there are very few examples of quantum homogeneous anti de Sitter spaces. We expect that the results of this section generalise to 3d de Sitter and Minkowski space, although the concrete description of the {coisotropic}  Poisson homogeneous structures will have to rely on a different set of coordinates.

\subsection{Lorentzian 3d constant curvature spacetimes as  Poisson homogeneous spaces}
\label{subsec:colg}

The  Lie algebras $\gothg_\Lambda=\mathrm{Lie}\, G_\Lambda$ associated with the isometry groups of 3d gravity  can again be described in a unified way, in terms of a  a basis that is adapted to the 3d Cayley-Klein geometries.  It consists of generators  $\{J,P_0,P_1,P_2,K_1,K_2\}$ with the Lie bracket
\be
\begin{array}{lll} 
[J,P_i]=   \epsilon_{ij}P_j , &\qquad
[J,K_i]=   \epsilon_{ij}K_j , &\qquad  [J,P_0]= 0  , \\[2pt]
[P_i,K_j]=-\delta_{ij}P_0 ,&\qquad [P_0,K_i]=-P_i ,&\qquad
[K_1,K_2]= -J   \label{crules}, \\[2pt]
[P_0,P_i]=\kk K_i ,&\qquad [P_1,P_2]= -\kk J  ,
\end{array}
\ee
where $i,j=1,2$ and $\epsilon_{12}=-\epsilon_{21}=1$.  Depending on the cosmological constant $\Lambda=-\omega$, this yields the Lie algebras $\mathfrak{so}(2,2)\cong \mathrm{sl}(2,\R)\oplus \mathfrak{sl}(2,\R)$ for $\Lambda<0$, the Lie algebra $\mathfrak{so}(3,1)=\mathfrak{sl}(2,\mathbb C)$ for $\Lambda>0$ and the Lie algebra $\mathfrak{iso}(2,1)$ for $\Lambda=0$. In all cases,   $J$ generates a spatial rotation,  $P_0$ a time translation, the elements $K_i$ Lorentz boosts and the elements $P_i$ spatial translations. The Lie algebra $\gothg_\Lambda$ has two quadratic Casimir elements 
\be
{\cal C}=P_0^2-P^2+\kk(J^2-K^2), \qquad
{\cal W}=-JP_0+K_1P_2-K_2P_1  .
\label{bc}
\ee
The Casimir $\cal C$ corresponds to  the Killing--Cartan form and it is related to the energy of the particle, while $\cal W$ is related to  the Pauli-Lubanski vector.
Note that the basis $\{J,K_1,K_2,P_0,P_1,P_2\}$ is related to the basis $\{J_0,J_1,J_2, P_0,P_1,P_2\}$ considered in \cite{BHMNsigma, BHMplb1, BHMcqg, BHMplb2}
via a very simple change of basis \cite{BHMNsigma}
\begin{align}\label{eq:bchange}
&J=J_0 & &K_1=J_2 & &K_2=-J_1 & &P_i=P_i\;\qquad \text{for}\;i=0,1,2.
\end{align}
Our first aim is to determine all families of quasitriangular  Poisson-Lie structures on $G_\Lambda$ that are defined for all values of $\Lambda$ and for which the 3d Lorentz group   $\mathrm{SL}(2,\R)\subset G_\Lambda$ is a {\em Poisson subgroup}. This issue can be addressed in full generality  since
any such family of classical $r$-matrices for $r\in \gothg_\Lambda\oo\gothg_\Lambda$ must in particular determine a classical $r$-matrix  with this property for $\mathfrak{so}(2,2)$. The 
  classical $r$-matrices for  Lie algebra $\mathfrak{so}(2,2)$ were classified in~\cite{Tallin}. As any
 antisymmetric element of  $r\in \mathfrak{so}(2,2)\otimes \mathfrak{so}(2,2)$ can be parametrised as
\begin{eqnarray}
r\!&\!=\!&\!a_1 J \wedge P_1 + a_2 J\wedge K_1 + a_3 P_0\wedge P_1 + a_4 P_0 \wedge K_1 +a_5 P_1\wedge K_1 +a_6 P_1\wedge K_2 \cr 
 &&+\ b_1 J\wedge P_2 + b_2 J\wedge K_2 + b_3 P_0\wedge P_2 + b_4 P_0\wedge K_2 +b_5 P_2\wedge K_2 + b_6 P_2\wedge K_1 \label{genericr}\\ 
 &&+\  c_1 J\wedge P_0 + c_2 K_1\wedge K_2 + c_3 P_1\wedge P_2 ,
\nonumber
\end{eqnarray}
with 15 real parameters $a_1,...,a_6, b_1,...,b_6, c_1,...,c_3$, it is possible to explicitly derive the constraints on these parameters
that arise from the condition that $r$ is a solution of the modified classical Yang-Baxter equation (mCYBE). This was achieved in \cite{Tallin}. The conditions on these parameters that guarantee that the Lie algebra  $\mathfrak{sl}(2,\R)=\mathrm{Span}_\R\{J,K_1,K_2\}$ is a sub-{\em Lie bialgebra} of $\mathfrak{so}(2,2)$ can be derived by a straightforward computation from the  cocommutator defined by~\eqref{genericr}. On the Lie subalgebra $\mathfrak{sl}(2,\R)$, this cocommutator reads
\begin{eqnarray}
&&\back\mback  \delta(J) = a_1 J\wedge  P_2+ a_2 J\wedge  K_2 + a_3 P_0\wedge  P_2+ a_4 P_0\wedge  K_2 + (b_5 -a_5) (  K_1\wedge  P_2 - P_1\wedge  K_2) \nonumber\\
&&\back  +(a_6+b_6)( K_1\wedge  P_1 + P_2\wedge  K_2 ) + b_1 P_1\wedge  J + b_2K_1\wedge  J +  b_3 P_1\wedge  P_0 + b_4 K_1\wedge  P_0  ,\nonumber\\[0.1cm]
&&\back\mback  \delta(K_1)=  c_2  J\wedge K_1 + (a_1 + b_4) ( J\wedge P_0 + P_1\wedge K_2 ) + (a_6+ c_1) (J\wedge P_1+ P_0\wedge K_2 ) +  b_5   J\wedge P_2 \nonumber\\
&&\back + a_2 K_1\wedge K_2 +  a_5  P_0\wedge K_1 + c_3   P_0\wedge P_2+ a_4 P_1\wedge K_1+ b_3 P_1\wedge P_2 + b_1 P_2\wedge K_2  ,
\nonumber\\[0.1cm]\
&&\back \mback \delta(K_2)=  c_2   J\wedge K_2 + b_2 K_1\wedge K_2 + a_1 K_1\wedge P_1 + (a_4 -b_1)(  P_0\wedge J + P_2\wedge K_1 ) + b_5 P_0\wedge K_2  \nonumber \\
&&\back  +  ( b_6  - c_1) (P_0\wedge K_1+P_2\wedge J )+ a_5   P_1\wedge J +  c_3  P_1\wedge P_0 + b_4 P_2\wedge K_2 +a_3 P_2\wedge P_1,  \label{cocogen} 
\end{eqnarray}
and the condition $\delta(\mathfrak{sl}(2,\R))\subset \mathfrak{sl}(2,\R)\wedge \mathfrak{sl}(2,\R)$  is satisfied if and only if 
\be
a_1=a_3=a_4=a_5=b_1=b_3=b_4=b_5=c_3=0,
\qquad\qquad
b_6=c_1=-a_6.
\ee
If these equations hold, the mCYBE reduces to a single additional condition, namely
\be
a_2^2+b_2^2-c_2^2 + 4\,\Lambda\,a_6^2=0.
\label{constraint}
\ee
Therefore, the most general quasitriangular Poisson-Lie structure on $\mathrm{SO}(2,2)$ for which the 3d Lorentz group $\mathrm{SL}(2,\R)\cong \mathrm{SO}(2,1)$ is a Poisson-Lie subgroup is given by the $r$-matrix
\be 
r=a_2 J\wedge K_1 + b_2 J\wedge K_2 + c_2 K_1\wedge K_2 + a_6 \,(-J\wedge P_0  + P_1\wedge K_2 +  K_1\wedge P_2),
\label{genericpsc}
\ee
together with the condition~\eqref{constraint}. 
Note that for $a_6=0$ we obtain the $r$-matrix
\be
r=a_2 J\wedge K_1 + b_2 J\wedge K_2 + c_2 K_1\wedge K_2
\qquad\text{with}\qquad
a_2^2+b_2^2=c_2^2,
\label{carrier}
\ee
which is a solution of the classical Yang-Baxter equation. We thus obtain a {\em triangular} Poisson-Lie structure which is just the extension of the non-standard structure for $\mathfrak{sl}(2,\R)$ defined by~\eqref{parabollic}. 
Therefore, the only standard  solutions of the mCYBE have $a_6\neq 0$.
In particular, the family of  $r$-matrices  studied in~\cite{BHMcqg,BHMplb2}  that describes the Lie algebras $\gothg_\Lambda$ as classical  doubles of the Lie bialgebra $\mathfrak{sl}(2,\R)$ for all values of $\Lambda$ and is relevant for  3d gravity
satisfies these conditions. They are obtained by taking $b_2=c_2=0$, $a_2^2=\eta^2=-\Lambda$ and $a_6=-1/2$ in \eqref{carrier}, which yields
\be \label{eq:r1}
r=\eta\, J\wedge K_1  -\tfrac12 \,(-J\wedge P_0  - K_2\wedge P_1 +  K_1\wedge P_2)
\ee
and coincides with the family of  $r$-matrices in~\cite{BHMcqg,BHMplb2} up to the change of basis in \eqref{eq:bchange}.
Note also that the  parameter  $b_2$ can be taken to vanish without loss of generality due to the  classification of 
$r$-matrices for 
$\mathfrak{sl}(2,\R)$ discussed  in the previous section. The {coisotropic} Poisson homogeneous spaces for  this $r$-matrix will be discussed for the anti de Sitter case in the next subsection, where we show that they are an example of the double {coisotropic}  Poisson homogeneous spaces described in Section \ref{subsec:doubledouble}.

The general expression \eqref{genericr}
 for a classical $r$-matrix  also shows that the coisotropic subgroup condition is much weaker than the Poisson subgroup condition  and  that there are many ways of realising  $\mathrm{AdS}_3$ as a {\em coisotropic} Poisson homogeneous space over   $\mathrm{SO}(2,2)$  with a quasitriangular Poisson-Lie structure. In fact, by a direct inspection of~\eqref{cocogen} one finds that the {condition} $\delta(\mathfrak{sl}(2,\R))\subset \mathfrak{sl}(2,\R)\wedge \mathfrak{so}(2,2)$  leads to the conditions
\be
a_3=b_3=c_3=0,
\ee
together with the 18   quadratic constraints equations arising from the mCYBE \cite{Tallin}.
The set of quasitriangular  Poisson-Lie structures for which the group $\mathrm{SL}(2,\R)$ is only required to be {coisotropic therefore depends on 12 parameters, subject to  18 quadratic constraint equations}, and some of the solutions of these equations will lead to equivalent Poisson-Lie structures via automorphisms. 

An important example that arises from a classical $r$-matrix that gives the Lie algebra $\gothg_\Lambda$ the structure of a classical double is the 
 twisted (`space-like') $\kappa$-AdS Poisson Lie group from~\cite{BHMNsigma}  given by 
\be
r
=\tfrac{1}{2} (
-K_2 \wedge P_0 - J \wedge P_1) + \tfrac{1}{2}\,{\xi\,K_1\wedge P_2 }.
\label{tconp}
\ee
It is related to the family of $r$-matrices from \cite{BHMNsigma} by the change of basis \eqref{eq:bchange} and corresponds to the case where  the only non-vanishing parameters  in \eqref{genericr} are given by
\be
a_1=-b_4=-\frac12,
\qquad
b_6=-\frac{\xi}{2}.
\ee
Note that the associated Lie bialgebra structure coincides with the one of the usual $\kappa$-AdS Poisson-Lie group from \cite{BHOS3d,Starodutsev}  for $\xi=0$.
As  this is also a classical double $r$-matrix for the isometry group $\mathrm{SO}(2,2)$ \cite{BHMcqg} and due to the relevance of the $\kappa$-deformation, we will also  construct the corresponding {coisotropic}  Poisson homogeneous spaces and compare them to the one for the first family of $r$-matrices.

\subsection{$\mathrm{AdS}_3$ as a Poisson subgroup Poisson homogeneous space over a double}
\label{sec:3dgrav}
 In this and the following subsection we  realise  3d anti de Sitter space and via a limiting procedure also 3d Minkowski space as
 a Poisson homogeneous space over a double of a classical double. It can be expected that this will generalise straightforwardly also to 3d de Sitter space, but the concrete parametrisation of the Poisson structure in terms of coordinates will take a different form. 
 
We start  by considering the associated Lie bialgebra structures on $\mathfrak{so}(2,2)$.
As shown in~\cite{BHMcqg}, there  are two inequivalent families of classical $r$-matrices for $\mathfrak{so}(2,2)$ that give $\mathfrak{so}(2,2)$ the structure of a classical double $\mathfrak{so}(2,2)=D(\mathfrak{a})$ for a 3d  Lie bialgebra  $\mathfrak{a}$.
The first family is associated with the standard quantum deformation of $\mathfrak{a}=\mathfrak{sl}(2,\mathbb{R})$ from \eqref{standard}.
In terms of the standard basis $\{X_0,X_1, X_2\}$ of $\mathfrak{sl}(2,\R)$ and the associated dual basis $\{x^0,x^1,x^2\}$  of $\mathfrak{sl}(2,\R)^*$, the Lie brackets of the classical double read \cite{BHMcqg}
\begin{align}
&[X_0,X_1]= 2\,X_1 ,
&
&[X_0,X_2]=  -2\,X_2 ,
&
&[X_1,X_2]= X_0, \nonumber\\
&[x^0,x^1]= -\tfrac 1 2 \eta\,x^1 ,
&
&[x^0,x^2]=-\tfrac 1 2 \eta\,x^2,
&
&[x^1,x^2]=0,\nonumber\\
&[x^0,X_0]=0,
&
&[x^0,X_1]=x^2+\tfrac 1 2 \eta\,X_1 ,
& 
&[x^0,X_2]=-x^1+\tfrac 1 2 \eta\,X_2, \\
&[x^1,X_0]=2 x^1,
&
&[x^1,X_1]=-2 x^0-\tfrac 1 2 \eta\,X_0,
&
&[x^1,X_2]=0, \nonumber \\
&[x^2,X_0]=- 2 x^2,
&
&[x^2,X_1]= 0,
&
&[x^2,X_2]=2 x^0-\tfrac 1 2 \eta\,X_0,\nonumber
\end{align}
were $\Lambda=-\eta^2<0$. 
The basis $\{X^i,x_i\}$ is related to the basis $\{J,P_0,P_1,P_2,K_1,K_2\}$ introduced before \eqref{crules} by the basis transformation
\begin{align}
&J=-\tfrac12 (X_1 -X_2), & 
&K_1=\tfrac12 (X_1 +X_2)   ,&  
& K_2=-\tfrac12 X_0 ,\label{csbasis6}\\
&P_0=-\tfrac 1 2 \eta  (X_1 +X_2) +(x^1 - x^2),& 
&P_1=2 x^0,&  
&P_2=\tfrac 1 2 \eta  (X_1 -X_2) +(x^1 + x^2).
\nonumber
\end{align}
A direct computation shows that inserting these expressions into the Lie brackets above yields indeed the Lie bracket \eqref{crules} for the Lie algebra $\mathfrak{so}(2,2)$  with $\omega=\eta^2$ and transforms the canonical skew-symmetric $r$-matrix $r=\sum_i{x^i\wedge X_i}$ for $D(\mathfrak{sl}(2,\R))$ into the $r$-matrix \eqref{eq:r1}. 
The associated cocommutator reads
\begin{align}\label{cocomm}
&\delta(  J)=-\m     K_2\wedge   J,\qquad \delta(  K_2)=0,\qquad \delta(  K_1)=-\m    K_2\wedge   K_1 , &&\nonumber\\
&\delta(  P_0)= \left(    P_1\wedge  P_2 +\m     P_1\wedge   J-\m^2  K_1\wedge   K_2 \right) ,& &\\
&\delta(  P_1)= \left(       P_0\wedge   P_2+\m       P_0\wedge   J - \m    P_2\wedge  K_1+\m^2      K_1\wedge   J\right) ,& &\nonumber\\
&\delta(  P_2)= \left(        P_1\wedge  P_0   +\m       P_1\wedge   K_1 -\m^2       J\wedge  K_2\right) ,& &\nonumber
\end{align}
and this shows that the Lie algebra $\mathfrak{sl}(2,\R)=\mathrm{Span}\{J,K_1,K_2\}$ defines a {coisotropic}  Lagrangian Lie subalgebra  $\mathfrak l\subset D(\mathfrak{so}(2,2))=D(D(\mathfrak{sl}(2,\R)))$ of Poisson subgroup type. This is  consistent with the fact  that the classical $r$-matrix~\eqref{eq:r1} is a special case of the more general family of  $r$-matrices in~\eqref{carrier} that satisfy this condition. The Lagrangian Lie subalgebra $\mathfrak l\subset D(\mathfrak{so}(2,2))$ is given as $\mathfrak l=\mathrm{Span}_\R\{J,K_1,K_2, p^0,p^1,p^2\}$ with the Lie bracket
\begin{align}
[J, K_2]&=-K_1\,, & [J,K_1]&=K_2\,, & [K_1,K_2]&=-J\,,  \nonumber\\
[p^0,p^1]&=\,-p^2, & [p^0,p^2]&=\,p^1 & [p^1,p^2]&=\,p^0,
\label{linear}\\
[p^0,J]&=0\,, & [p^0,K_2]&=p^2\,, & [p^0,K_1]&=p^1\,, \nonumber\\
[p^1,J]&=-p^2\,, & [p^1,K_2]&=0\,, & [p^1,K_1]&=p^0\,, \nonumber  \\
[p^2,J]&=p^1\,, & [p^2,K_2]&=p^0\,, & [p^2,K_1]&=0\,.\nonumber
\end{align}
The Poisson structure on the homogeneous space $\mathrm{AdS}_3=\mathrm{SO}(2,2)/\mathrm{SL}(2,\R)$ is obtained from the Sklyanin bracket on $\mathrm{SO}(2,2)$
for the classical $r$-matrix \eqref{eq:r1}. This Sklyanin bracket was computed in~\cite{BHMplb2}. In terms of the  $\mathrm{AdS}_3$ group coordinates $x_a\equiv p^a$  the associated Poisson structure on $\mathrm{AdS}_3$ takes the form ~\cite{BHMplb2}
\bea
&&   \{x_0,x_1\}  = -\frac{\tanh\m x_2 }{\m} \,\C(x_0,x_1)=- x_2 + o[\m] ,\,  \nonumber\\
&& \{x_0,x_2\} =  \frac{ \tanh\m x_1}{\m}\,\C(x_0,x_1)\, = x_1 + o[\m] ,\\
&&  \{x_1,x_2\} = \frac{\tan\m x_0}{\m}\,\C(x_0,x_1)\, =  x_0 
+ o[\m] ,
\nonumber
\eea
where the cosmological constant is given by $\Lambda=-\eta^2$ and the function $\C$ reads
\be
\C(x_0,x_1)\equiv\cos\m x_0(\cos\m x_0\cosh\m x_1+ \sinh\m x_1)\,.
\ee
Note that the linearisation of this Poisson bracket  is just the Lie bracket~\eqref{linear} and does not vanish, in contrast to the corresponding bracket $\mathrm{AdS}_2$ from Table 1.  In fact, it generates a Poisson algebra whose quantisation is by no means trivial, as discussed in~\cite{BHMplb2}. Note  that this Poisson structure is well-defined also in the flat limit $\eta\to 0$, i.~e.~the case of vanishing cosmological constant, and this describes the associated {coisotropic}  Poisson homogeneous structure on 3d Minkowski space. 

\subsection{$\mathrm{AdS}_3$ as a coisotropic Poisson homogeneous space over a double}
\label{subsec:coisot}

In this subsection we investigate the second family of classical $r$-matrices which give the Lie algebra $\mathfrak{so}(2,2)$ the structure of a classical double $\gothg=D(\mathfrak{a})$, where $\mathfrak a=\mathfrak{iso}(1,1)$ is  the Poincar\'e algebra  in two dimensions with a bialgebra structure that  depends on an essential deformation parameter $\eta\neq 0$. In terms of a basis $\{X_0,X_1,X_2\}$ of $\mathfrak{iso}(1,1)$ and the associated dual basis $\{x^0,x^1,x^2\}$ of $\mathfrak{iso}(1,1)^*$ the Lie bracket of $\gothg_\Lambda=D(\mathfrak{iso}(1,1))$ is given by~\cite{BHMcqg}
\begin{align}
&[X_0,X_1]= - X_2,
&
&[X_0,X_2]=  -X_1,
& 
&[X_1,X_2]=0,  \nonumber\\
&[x^0,x^1]= \eta\,x^1 ,
&
&[x^0,x^2]=\eta\,x^2,
&
&[x^1,x^2]=0,
\nonumber \\
&[x^0,X_0]=0,
&
&[x^0,X_1]=-\eta\,X_1 ,
&
&[x^0,X_2]=-\eta\,X_2, \label{zj}\\
&[x^1,X_0]=-x^2,
&
&[x^1,X_1]=\eta\,X_0,
&
&[x^1,X_2]=x^0, \nonumber\\
&[x^2,X_0]=-x^1,
&
&[x^2,X_1]= x^0,
&
&[x^2,X_2]=\eta\,X_0,
\nonumber
\end{align}
where $\Lambda=-\eta^2<0$ and $\eta>0$.
If we  perform the change of basis 
\begin{align}
&J=\frac{1}{\sqrt{2\eta}}(X_2 - x^1) ,&  
&K_1=- \frac 1 \eta x^0, & 
&K_2=-\frac{1}{\sqrt{2\eta}} (X_2 + x^1)  ,\nonumber\\
&P_0=\sqrt{ \frac{\eta}{{2}}} (X_1 - x^2),&  
&P_1=\sqrt{\frac{\eta}{{2}}} (X_1 + x^2), &
& P_2=-\eta X_0,  \label{csbasis7}
\end{align}
we obtain the Lie algebra $\mathfrak{so}(2,2)$  with the bracket \eqref{crules}. 
A similar change can be employed in the case where $\eta<0$. In terms of the basis $\{J,K_1,K_2,P_0,P_1,P_2\}$ the
 antisymmetric canonical $r$-matrix of $D(\mathfrak{iso}(1,1))$ then takes the form \eqref{tconp} with $\xi=1$.
The associated quantum group  is a superposition of the (space-like) $\kappa$-deformation of $\mathfrak{so}(2,2)$ from \cite{BHOS3d} that is generated by  the term $(-K_2\wedge P_0 - J \wedge P_1)$ and a Reshetikhin twist  generated by $K_1\wedge P_2$.  
Anti de Sitter space $\mathrm{AdS}_3$ is again given as a {coisotropic}  Poisson homogeneous space  $\mathrm{AdS}_3=\mathrm{SO}(2,2)/\mathrm{SL}(2,\R)$ and the subgroup $\mathrm{SL}(2,\R)\subset \mathrm{SO}(2,2)$ is again generated by $J,K_1,K_2$. However, for this Lie bialgebra structure the cocommutator is given by
\begin{align}
\label{cocomm}
\delta(J)&=\tfrac{1}{2}(P_2\wedge J-P_2\wedge K_2+K_1\wedge P_0+P_1\wedge K_1)\,, \nonumber \\
\delta(K_2)&=\tfrac{1}{2}(P_2\wedge J-P_2\wedge K_2-K_1\wedge P_0-P_1\wedge K_1)\,, \nonumber \\
\delta(K_1)&=0\,, \nonumber \\
\delta(P_0)&=\tfrac{1}{2}\left(\eta ^2(K_1\wedge J-K_2\wedge K_1)+P_2\wedge (P_0+P_1)\right)\,,\\
\delta(P_1)&=\tfrac{1}{2}\left(\eta ^2(-K_1\wedge K_2+J\wedge K_1)+P_2\wedge (P_0+P_1)\right)\,, \nonumber \\
\delta(P_2)&=0\,. \nonumber
\end{align}
This shows that $\mathrm{sl}(2,\R)\subset \mathfrak{so}(2,2)$ is {a coisotropic Lagrangian Lie subalgebra, but not a sub-{\em Lie bialgebra}}, as discussed in Section \ref{subsec:colg}. The associated Lagrangian subalgebra of $D(\mathfrak{so}(2,2))=D(D(\mathfrak{iso}(1,1)))$ is again of the form $\mathfrak l=\mathrm{Span}_\R\{J,K_1,K_2,p^0,p^1,p^2\}$
with Lie bracket
\begin{align}
[J, K_2]&=-K_1\,, & [J,K_1]&=K_2\,, & [K_1,K_2]&=-J\,,  \nonumber\\
[p^0,p^1]&=0\,, & [p^0,p^2]&=-\tfrac{1}{2}(p^0+p^1)\, & [p^1,p^2]&=-\tfrac{1}{2}(p^0+p^1)\,, 
 \label{linkap}\\
[p^0,J]&=\tfrac{1}{2}K_1\,, & [p^0,K_2]&=p^2+\tfrac{1}{2}K_1\,, & [p^0,K_1]&=p^1\,, \nonumber \\
[p^1,J]&=-\left(p^2+\tfrac{1}{2}K_1\right)\,, & [p^1,K_2]&=-\tfrac{1}{2}K_1\,, & [p^1,K_1]&=p^0\,,\nonumber\\   
[p^2,J]&=\tfrac 1 2 (J-K_2) 
& [p^2,K_2]&=p^0+\tfrac{1}{2}(J-K_2)\,, & [p^2,K_1]&=0\,. \nonumber
\end{align}
Hence, we obtain a {coisotropic}  Lagrangian subalgebra $\mathfrak l$
 of the general form~\eqref{coiscr}, 
 but in contrast to the sub-Lie bialgebra case, it is {\em not} a semidirect product of the Lie algebras $\mathfrak l\cap \mathfrak{so}(2,2)=\mathfrak{sl}(2,\R)$
and the Lie algebra $\mathfrak l\cap \mathfrak{so}(2,2)^*$.

The Poisson bracket on the homogeneous space $\mathrm{AdS}_3=\mathrm{SO}(2,2)/\mathrm{SL}(2,\R)$ is again obtained from the Sklyanin bracket on $\mathrm{SO}(2,2)$,  in this case for  the classical $r$-matrix \eqref{tconp} for $\xi=1$. However, to identify  the contributions of each part of the $r$-matrix, it is convenient to consider the $r$-matrix \eqref{tconp} for general values of the parameter $\xi$, which controls the contribution of the twist $K_1\wedge P_2$. A convenient parametrisation of this Sklyanin bracket in suitable coordinates  on $\mathrm{SO}(2,2)$ was computed in~\cite{BHMplb2}. The projection of this bracket to the group coordinates $x_a\equiv p^a$ that parametrise $\mathrm{AdS}_3$  then gives the following bracket on $\mathrm{AdS}_3$
\begin{align}
\{x_0,x_1\}&={\frac{\xi}{2}\,\,\frac{\tanh \eta x_2}{\eta}\,\sech \eta x_1\left(\cos ^2\eta x_0\,\sinh ^2\eta x_1-\sin ^2\eta x_0\right)=  o[\m^2]} , \nonumber\\
 \{x_0,x_2\}&=-\tfrac{1}{2}\,\,
  \frac{\sin \eta x_0}{\eta}\,\cosh\eta x_1 + \frac{\sinh \eta x_1}{2\,\eta}\left(\sin\eta x_0\, \tanh \eta x_1 {- \xi\cos^2\eta x_0}
  \right)
 \nonumber 
= -\tfrac{1}{2}(x_0+{\xi\,x_1}) + o[\m^2] , \nonumber \\
\{x_1,x_2\}&=-\tfrac{1}{2}\,\,\frac{\sinh \eta x_1}{\eta}\,\cos\eta x_0
{-\frac{\xi}{2}\,\frac{\sin \eta x_0}{\eta}\,\cos\eta x_0\,\cosh\eta x_1}
= -\tfrac{1}{2}({\xi\,x_0}+x_1)  + o[\m^2] ,\label{eq:poissother} 
\end{align}
where again   $\Lambda=-\eta^2$.
The linearisation of this Poisson bracket reads
\begin{align}\label{eq:linb}
\{x_0,x_1\}&=0, & 
 \{x_0,x_2\}&= -\tfrac{1}{2}(x_0+{\xi\,x_1}) ,   &
\{x_1,x_2\}&=-\tfrac{1}{2}({\xi\,x_0}+x_1) , 
\end{align}
and coincides with the Lie bracket~\eqref{linkap} for $\xi=1$. This is the flat   limit  $\eta\to 0$ of the bracket \eqref{eq:poissother}, which describes the case of vanishing cosmological constant and hence the {coisotropic}  Poisson homogeneous space $M_3$.   In the untwisted case, i.~e.~for $\xi=0$,  the flat limit $\eta\to 0$ of the bracket \eqref{eq:poissother} yields precisely the well-known (2+1) $\kappa$-Minkowski spacetime~\cite{kMas, kMR, kZakr}, which is indeed a coisotropic Poisson homogeneous space. 



\section{Concluding remarks}

 This article presents the first steps towards the construction of  the {coisotropic} quantum  homogeneous spaces associated with 
the three two-dimensional homogeneous spaces over the group $\mathrm{SL}(2,\R)$, including in particular 2d hyperbolic space $\mathbb H^2$ and anti de Sitter space $\mathrm{AdS}_2$,  as well as three-dimensional anti de Sitter space $\mathrm{AdS}_3$ and its flat limit.
The explicit construction of the full quantum homogeneous spaces is challenging, since it requires an explicit description of the  Hopf algebra duality between the quantum algebra and the quantum group. A precise understanding of the corresponding {coisotropic} Poisson homogeneous spaces and  the associated {coisotropic} Lagrangian subalgebras of the classical doubles is required and can provide helpful guidance in this task. The former can be viewed as the semiclassical limit  of the associated quantum homogeneous space, the latter as its infinitesimal version or first order approximation in $\hbar$.
It contains all essential information on the associated quantum homogeneous space and in principle allows one to construct it via an iterative procedure, order for order in $\hbar$.

The two-dimensional examples studied in this article are {important ones} in their own right due to the rich geometry of these classical homogeneous spaces and as toy models for the higher dimensional case. 
On the other hand, the 3d examples here presented are motivated by their role in 3d gravity. In this setting, the quantum group symmetries of quantum homogeneous spaces arise from Poisson-Lie symmetries in the classical theory and the relevant Poisson-Lie symmetries have been identified in \cite{cm2,cmn} as classical doubles. While the case of vanishing cosmological constant is structurally simpler and well understood, 3d anti de Sitter and de Sitter space are of special interest, since they allow one to investigate both, the cosmological constant and Planck's constant as deformation parameters  and allow one to understand the role of curvature. 

It would also be interesting to explore in more detail the Minkowski and de Sitter counterparts of the results in this article, both from the point of view of Poisson-Lie group contractions~\cite{CGST2, BHOSpl,ref} and via the classification of the quasitriangular Poisson-Lie structures on the groups $\mathrm{ISO}(2,1)$ and $\mathrm{SL}(2,\mathbb C)$. In particular, this raises the question if the unified description of the isometry groups in 3d gravity in terms of (pseudo) quaternions from \cite{cm3} can be used to give unified coordinates for the associated homogeneous spaces in which the cosmological constant appears as a deformation parameter. In this context, it would also be worthwhile to explore the connections between Poisson homogeneous spaces and the dynamical Yang-Baxter equation~\cite{KS,Ludrm}. It was shown in \cite{cmt,bn,br} that gauge fixing in the context of 3d gravity is related to the introduction of an observer in the theory  and leads to the appearance of {\em dynamical} $r$-matrices.

Finally, the analysis in this article can be viewed as a first step towards an analogous description of 4d constant curvature spacetimes. 
Like their three-dimensional counterparts,  four dimensional de Sitter, anti de Sitter and   Minkowski space are obtained as homogeneous spaces of their isometry groups with the Lorentz group $\mathrm{SO}(3,1)$ as the stabiliser subgroup. Hence, a similar analysis could be performed 
 to determine  {the coisotropic} Poisson homogeneous structures on these four-dimensional spacetimes.
In particular, the classical $r$-matrices for $\mathfrak{so}(3,1)$ were classified in~\cite{Zakrlorentz,Tolstoy}, see also~\cite{BLT}, and it would be interesting to determine  which of them generates Poisson subgroup homogeneous spaces on four dimensional de Sitter space $\mathrm{AdS}_4$.
 On the other hand, it was recently shown in~\cite{BHNplb} 
  that the Lie algebra of the isometry group of $\mathrm{AdS}_4$ can be realised as a classical double and that the  associated canonical $r$-matrix is a four-dimensional generalisation of the twisted kappa deformation~\eqref{tconp}. This yields a coisotropic Poisson homogeneous structure on $\mathrm{AdS}_4$ that is a direct generalisation of the coisotropic Poisson homogeneous structure on $\mathrm{AdS}_3$  studied in this article.
  It would be interesting to analyse the four-dimensional situation in more depth and to obtain more relevant examples of {coisotropic} Poisson homogeneous structures.


\begin{appendix}

\section{Coordinate functions and vector fields for the 2d Cayley-Klein geometries}
\label{sec:appa}

The fundamental representation $D$
of the Cayley-Klein Lie algebra $\mathfrak g_{(\k_1,\k_2)}$ with $\k_1=1$ and $\k_2=-1$ is
given by (see~\cite{BHOSpl} for details)
\begin{align}
D(P_1)=\begin{pmatrix}
0&-1&0\\1&0&0\\0&0&0\end{pmatrix},\, D(P_2)=\begin{pmatrix}
0&0&1\\0&0&0\\1&0&0\end{pmatrix} , \, D(J_{12})=\begin{pmatrix}
0&0&0\\0&0&1\\0&1&0\end{pmatrix}.
\end{align}
A real 
representation $D$ of the  group element $g=e^{a_1 P_1}e^{a_2
P_2}e^{\theta J_{12}}\in G_{(\k_1,\k_2)}$ is thus given by 
\begin{align}\label{eq:coordck}
&D(g)=\\
&\begin{pmatrix}
\CC(a_1) \CC(a_2) &\!\! -\k_1 \SS (a_1) \CC(\theta) -  \k_1\k_2
\CC(a_1)
\SS(a_2)\SS(\theta)&\!\!\k_1\k_2\SS(a_1)\SS(\theta)- \k_1\k_2 
\CC(a_1) \SS(a_2)\CC(\theta) \\
\SS(a_1) \CC(a_2) & \CC(a_1) \CC(\theta) - \k_1\k_2
\SS(a_1)
\SS(a_2)\SS(\theta) & - \k_2 \CC(a_1)\SS(\theta)- \k_1\k_2 
\SS(a_1) \SS(a_2)\CC(\theta) \\
\SS(a_2) & \CC(a_2) \SS(\theta)  & 
\CC(a_2)\CC(\theta) 
\end{pmatrix},\nonumber
\end{align}
where we use the abbreviations
$\CC(a_1)\equiv \cos (a_1)$, 
$\CC(a_2)\equiv\cosh(a_2)$,  $\CC(\theta)\equiv\cosh(\theta)$, $\SS(a_1)\equiv \sin (a_1)$,
$\SS(a_2)\equiv\sinh(a_2)$ and $\SS(\theta)\equiv\sinh(\theta)$.
In terms of these coordinates $\theta$, $a_1$, $a_2$, the left and right invariant vector fields
 for $\mathrm{SL}(2,\R)$ associated with the basis $\{P_1,P_2,J_{12}\}$ of $\mathfrak{sl}(2,\R)$ are given by
\begin{align}\label{eq:ckvecs}
{J_{12}}^L&=\partial_\theta,\\
{P_1}^L&={1\over{\cosh (a_2)}}\left\{ -
\sinh (a_2)\cosh (\theta)  \partial_\theta + \cosh (\theta)
\partial_{a_1} + \cosh (a_2)\sinh (\theta)
\partial_{a_2}\right\},\nonumber\\
{P_2}^L&={1\over{\cosh (a_2)}}\left\{-
\sinh (a_2)\sinh (\theta)  \partial_\theta + \sinh (\theta)
\partial_{a_1} + \cosh (a_2)\cosh (\theta)
\partial_{a_2}\right\}, \nonumber
\end{align}
\begin{align}\label{eq:ckvecs2}
\! {J_{12}}^R &={1\over{\cosh (a_2)}}\left\{ 
\sinh (a_2)\cos(a_1) \partial_{a_1} + \cos(a_1) 
\partial_{\theta} +\!
\cosh (a_2)\sin(a_1) \partial_{a_2}\right\},\\
{P_1}^R &=\partial_{a_1},\nonumber\\
\!\!\!\! {P_2}^R &={1\over{\cosh (a_2)}}\left\{ -
\! \sinh (a_2)\sin(a_1)  \partial_{a_1} \! -
\sin(a_1) \partial_{\theta} + \! \cosh (a_2)\cos(a_1)
\partial_{a_2}\right\}.\nonumber
\end{align}
On the other hand,  the two-dimensional representation  of $\mathfrak{sl}(2,\mathbb R)\simeq so(2,1)$
\be
 J_3=\left(\begin{array}{cc}
1&0\\0&-1\end{array}\right) ,\qquad
 J_+=\left(\begin{array}{cc}
0&1\\0&0\end{array}\right) ,\qquad 
 J_-=\left(\begin{array}{cc}
0&0\\1&0\end{array}\right) ,
\ee
allows one to parametrise  elements of the group ${\rm SL}(2,\RR)$ near the unit element as
\be
T= \ea^{a_-J_-}\ea^{a_+J_+}\ea^{\chi J_3}= \left(\begin{array}{cc} \ea^{\chi} & a_+\,\ea^{-\chi}\\
a_-\,\ea^{\chi} & (1 + a_-\,a_+)\ea^{-\chi} \end{array}\right) \equiv \left(\begin{array}{cc} a & b\\
c & d \end{array}\right),\qquad a d- bc =1.
\label{kkb}
\ee
In terms of the local coordinates
 $a_+$, $a_-$, $\chi$, the left and right invariant vector fields of ${\rm SL}(2,\RR)$  associated with the basis $\{J_3,J_\pm\}$ of $\mathfrak{sl}(2,\R)$ are given by~\cite{BHMNsigma}
\bea\label{kkc}
&& Y_{J_+}^L=\ea^{2\chi}\,\partial_{a_+},\\  
&& Y_{J_-}^L=a_+^2\,\ea^{- 2\chi}\,\partial_{a_+}+ \ea^{-
2\chi}\,\partial_{a_-} + a_+\,\ea^{- 2\chi}\,\partial_{\chi},\nonumber\\ 
&&  Y_{J_3}^L=\partial_{\chi}, \nonumber\\
&& Y_{J_+}^R =(1+2 a_- a_+)\,\partial_{a_+}-  a_-^2
\,\partial_{a_-} +  a_-\,\partial_{\chi}, \label{kkc2}\\ 
&&  Y_{J_-}^R = \partial_{a_-}  , \nonumber\\ 
&&  Y_{J_3}^R =- 2\,a_-\,\partial_{a_-}+ 2 a_+\,
\partial_{a_+}+\partial_{\chi} .\nonumber
\nonumber
\eea
\end{appendix}


\section*{Acknowledgements}

This work was partially supported by the Spanish MINECO under grants MTM2013-43820-P {and MTM2016-79639-P  (AEI/FEDER, UE)}, by Junta de Castilla y Le\'on (Spain) under grants BU278U14 and VA057U16, {and by the Action MP1405 QSPACE from the European Cooperation in Science and Technology (COST).}
P.~Naranjo acknowledges a postdoctoral fellowship from Junta de Castilla y Le\'on. {The authors are indebted to the referee for comments improving the paper and eliminating mistakes in the presentation of the background material.}


{

\end{document}